\documentclass[12pt]{article}
\usepackage{amscd, amsfonts, amsmath, amssymb, amstext, amsthm, caption, epsfig,fancyhdr, float, graphicx, latexsym, multicol, multirow,tikz}
\usepackage[lofdepth,lotdepth]{subfig}
\usepackage{lmodern}
\usepackage{authblk}
\usepackage[utf8]{inputenc}
\usepackage[T1]{fontenc}
\usepackage[hyphens,spaces,obeyspaces]{url}
\usepackage{hyperref}
\usepackage{gensymb}
\usepackage{relsize}
\usepackage{booktabs}
\usepackage{mathtools}

\newcommand{\vh}{\widehat{V}}
\newcommand{\ph}{\widehat{\psi}}
\usepackage{tkz-graph}

\usetikzlibrary{backgrounds,automata}

\makeatletter

\usepackage{color}
\newtheorem{lemma}{Lemma}[section]
\newtheorem{theo}{Theorem}[section]
\newtheorem{coro}{Corollary}[section]
\newtheorem{prop}{Proposition}[section]

\newtheorem{rem}{Remark}[section]

\newcommand \E{\mathbb{E}}
\newcommand \Esp[1]{\Esp\left(#1\right)}
\renewcommand \P{\mathbb{P}}

\usepackage{xspace}
\newcommand{\iid}{{i.i.d.}\ }

\newcommand{\pmf}{{p.m.f.}\ }

\newcommand{\N}{\mathbb{N}\xspace}

\begin{document}

\title{On the profitability of selfish blockchain mining\\ under consideration of ruin}
\author{Hansj\"{o}rg Albrecher\footnote{\footnotesize Department of Actuarial Science, Faculty of Business and Economics, University of Lausanne and Swiss Finance Institute. Batiment Extranef, UNIL-Dorigny, 1015 Lausanne, Switzerland.\\ Email: hansjoerg.albrecher@unil.ch}\,\,  and Pierre-O. Goffard\footnote{ISFA, Université Lyon 1, LSAF EA2429. Email: pierre-olivier.goffard@univ-lyon1.fr}}

\date{}
\maketitle
\vspace{3mm}
\begin{abstract}
{Mining blocks on a blockchain equipped with a proof of work consensus protocol is well-known to be resource-consuming. A miner bears the operational cost, mainly electricity consumption and IT gear, of mining, and is compensated by a capital gain when a block is discovered. This paper aims at quantifying the profitability of mining when the possible event of ruin is also considered. This is done by formulating a tractable  stochastic model and using tools from applied probability and analysis, including the explicit solution of a certain type of advanced functional differential equation. The expected profit at a future time point is determined for the situation when the miner follows the protocol as well as when he/she withholds blocks. The obtained explicit expressions allow to analyze the sensitivity with respect to the different model ingredients and to identify conditions under which selfish mining is a strategic advantage.}
\end{abstract}

\section{Introduction}\label{sec:introduction}
A blockchain is a distributed public data ledger maintained by achieving consensus among a number of nodes in a \textit{Peer-to-Peer} network. The nodes, referred to as miners, are responsible for the validity of the information recorded in the blocks. Each miner stores a local copy of the data and competes to solve a cryptographic puzzle. The first miner who is able to propose a solution includes the pending information in a block and collects a reward.  The mining process is energy-consuming and costly in terms of equipment. The goal of the present paper is to provide insights into how profitable it is to engage in mining activities when keeping the aforementionned expenses in mind.\\

\noindent Miners are supposed to follow an incentive-compatible protocol which consists of releasing immediately any newly discovered block and appending it to the longest existing branch of the blockchain. By doing so, miners receive reward proportionally to their contribution to the overall computing effort. It was shown by Kroll et al. \cite{KrDaFe13} that following the protocol is a dominant strategy leading to a consistent history of transactions. Later, Eyal and Sirer \cite{EySi18} introduced a blockwithholding strategy, called \textit{selfish mining}, that allows a pool of miners to get a revenue, relative to the network revenue, that is greater than its fair share of the computing power. Selfish mining entails a waste of resources for all the network participants but makes the honest nodes waste more. Other blockwithholding strategies were presented in Sapirshtein et al.\  \cite{SaSoZo16} and Nayak et al.\  \cite{NaKuMiSh16} leading to an optimized relative revenue. These strategies imply a lower revenue, even for the malicious nodes, which led Grunspan and Marco-P\'erez in a series of contributions \cite{GrPM19,GrPM18selfish,GrPM18Trailing,GrPM18Stubborn} to label them as not profitable. The authors pointed out that the time required to implement such a scheme can be rather long and the time before it becomes profitable even longer exposing them to a significant risk of going bankrupt. Blockchain protocols usually calibrate the difficulty of the cryptopuzzle to ensure a steady flow of confirmed information. The difficulty is adjusted periodically to account for the evolution of the computing power of the network. Because blockwithholding strategies impede the speed of the block generating process, their application will cause a decrease in the cryptopuzzle difficulty for the next round. The selfish miners resume then to follow the protocol and increase their revenue. Furthermore, the waste of resources will drive some honest miners out of the game, increasing the colluding miners' share of the network computing power. Yet another side effect of selfish mining, noted in Eyal and Sirer \cite{EySi18}, is that the revenue of the selfish miners will entice honest miners to join them. The pool will grow and may accumulate more than $50\%$ of the resources putting them in position to perform double spending attacks, described in Nakamoto \cite{Na08}, Rosenfeld \cite{Ro14} and Goffard \cite{Go19}.\\

\noindent In this paper we model the surplus of a miner by a stochastic process that is inspired by risk processes considered for the study of the surplus in insurance portfolios, for which a detailed understanding concerning the probability of that surplus to become negative is available (see e.g.\ Asmussen and Albrecher \cite{AsAl10} for an overview). The latter event is referred to as {\em ruin} in the insurance context, and we will adopt this terminology in the present setup, even if the investigation of the event of running out of money can be seen as a risk management tool rather than bankruptcy in the literal sense. The profitability will be assessed by computing the expected surplus of a miner at some time horizon given that ruin did not occur until then. It will turn out that under reasonably realistic model assumptions, one can derive closed form expressions for this quantity when letting the time horizon to be random, and more in particular, exponentially distributed. Apart from tractability, the assumption of an exponential time horizon also has a natural interpretation in terms of lack-of-memory as to when the surplus process will be observed for the assessment of the profitability. Note that the selfish mining strategy considered in this paper is a simplified version of the one studied in Eyal and Sirer \cite{EySi18}, with the purpose to allow an explicit treatment for the value function under consideration. \\


\noindent The rest of the paper is organized as follows. Section \ref{sec:PnL_miner} develops a stochastic model for the surplus process of a miner and derives formulas for safety and profitability measures for a miner that follows the prescribed protocol. Mathematically, our profitability analysis will lead to the solution of differential equations with advanced arguments. In Section \ref{sec:blockwithholding_strategy}, the study is then extended to the situation of a selfish miner. An adaptation of the definition of selfish mining allows to make the surplus process Markovian (in a suitable higher-dimensional space), and the resulting system of differential equations with advanced arguments is analytically solved in an appendix. Section \ref{sec:numerical_illustrations} subsequently uses the obtained explicit expressions for numerical studies and sensitivity tests with respect to model parameters in order to determine whether blockwithholding strategies are profitable or not. We find that selfish mining may be profitable,  but always to a lesser extent than following the protocol. In Section \ref{sec:diffi}, we consider a time horizon large enough to include a difficulty adjustment. We show that, depending on the electricity price, in the presence of difficulty adjustments selfish mining may in fact outperform the alternative of following the protocol. In Section \ref{sec:conc} we conclude and outline some directions for future research. 

\section{Ruin and expected surplus of a miner following the protocol}\label{sec:PnL_miner}
A miner, referred to as Sam, starts operating with an initial surplus level $u>0$. We consider the following stochastic model for the surplus process over time. Mining blocks generates steady operational costs of amount $c>0$ per unit of time, most notably due to electricity consumption. The entire network of miners appends blocks to the blockchain at an exponential rate $\lambda$, which means that the length of the blockchain over time is governed by a Poisson process $(N_t)_{t\geq0}$ with intensity $\lambda$. We assume that Sam owns a fraction $p\in(0,1/2)$ of the overall computing power, which implies that each block is published by Sam with a probability $p$. The number of blocks found by Sam and therefore the number of rewards of size $b>0$ he collects up to time $t\geq 0$ is a thinned Poisson process $(\tilde{N}_t)_{t\geq0}$ with intensity $\lambda$ and thinning parameter $p$. Sam's surplus $R_t$ at time $t$ is then given by  
\begin{equation}\label{eq:surplus_protocol}
R_t = u + b\cdot\tilde{N}_t - c\cdot t,\qquad t\ge 0.
\end{equation}
The stochastic process $(R_t)_{t\geq0}$ resembles the so-called dual risk process in risk theory, see for instance Avanzi et al.\ \cite{AVANZI2007111}. The focus of this paper is the profitability of mining blocks on the blockchain, but subject to a ruin constraint. In this context, the time of ruin $\tau_u = \inf\{t\geq0:\text{ }R_t = 0\}$ is defined as the first time when the surplus reaches zero, i.e.\ the miner runs out of money and cannot continue to operate. The riskiness of the mining business may be assessed via the finite-time and infinite-time horizon ruin probabilities defined as 
\begin{equation}\label{eq:ruin_probabilities}
\psi(u,t) = \mathbb{P}(\tau_u\leq t),\text{ and }\psi(u)=\mathbb{P}(\tau_u<\infty), 
\end{equation}
respectively. The rewards earned through mining must in expectation exceed the operational cost per time unit. The latter condition translates into $p\lambda b > c$, and is referred to as the net profit condition in standard risk theory, see \cite{AsAl10}. In particular, this implies $\psi(u)<1$, i.e.\ ruin does not occur almost surely. The profitability for a time horizon $t>0$ is now measured by
\begin{equation}\label{eq:value_function_finite_time}
V(u,t) = \E(R_t\cdot \mathbb{I}_{\tau_u>t}),
\end{equation}
where $\mathbb{I}_A$ denotes the indicator random variable of an event $A$. Correspondingly, $V(u,t)$ is the expected surplus level at time $t$, where in the case of ruin up to $t$ this surplus is 0 (i.e., due to the ruin event the surplus is frozen in at 0). In terms of a conditional expectation, one can equivalently express $V(u,t)$ as the probability to still be alive at $t$ times the expected value of the surplus at that time given that ruin has not occurred:
\[V(u,t)= (1-\psi(u,t))\cdot \E(R_t\vert {\tau_u>t}).\]
 The following result provides formulas for the ruin probabilities and $V(u,t)$ in this setting.
\begin{prop}\label{prop:ruin_proba_and_value_func} 
\text{ }
\begin{enumerate}
\item For any $u\ge 0$, the finite-time ruin probability is given by 
\begin{equation}\label{eq:finite_time_ruin_proba}
\psi(u,t) = \sum_{n = 0}^{\infty}\frac{u}{u+bn}\;\P\left[N_{\frac{u+bn}{c}} = n\right]\mathbb{I}_{\left\{t>\frac{u+bn}{c}\right\}}. 
\end{equation}
\item For any $u\ge 0$, the infinite-time ruin probability is given by
\begin{equation}\label{eq:infinite_time_ruin_proba}
\psi(u) =e^{-\theta^\ast u},
\end{equation}
where $\theta^\ast$ is the positive solution in $\theta$ of the equation 
\begin{equation}\label{lueq}{c}\,\theta + p\lambda \, (e^{-b\, \theta }-1)=0.\end{equation}
\item For any $u\ge 0$, the expected surplus at time $t$ in case ruin has not occurred until then, can be written as 
\begin{equation}\label{eq:value_function_finite_time_prop}
V(u,t) = \E\left[\left(u+bN_t - ct\right)_+(-1)^{N_t}G_{N_t}\left(0\;\Big\rvert \left\{\frac{u}{ct}\land 1,\ldots, \frac{u+(N_t-1)b}{ct}\land 1\right\}\right) \right],
\end{equation}
where $(.)_+$ denotes the positive part, $\land$ stands for the minimum operator and\\ $\left(G_n(\cdot\rvert\{\ldots\}\right)_{n\in\N}$ is the sequence of Abel-Gontcharov polynomials as defined in Goffard and Lefevre \cite[Proposition 2.4]{GoLe17JMAA}.
\end{enumerate}
\end{prop}
\begin{proof}
\begin{enumerate}
\item The ruin time $\tau_u$ may be rewritten as 
\begin{equation*}
\tau_u = \inf\left\{t\geq 0\text{ ; }\tilde{N}_t = {ct}/{b} - {u}/{b}\right\}.
\end{equation*}
Note that ruin can only occur at the specific times 
$$
t_k = \frac{u+bk}{c}\text{, }k \geq0,
$$
when the function $t\mapsto {ct}/{b} - {u}/{b}$ reaches integer levels. For $t>0$, define the set of indices $\mathcal{I} = \{k\geq0\text{ ; }t_k\leq t\}$. The finite-time ruin probability can then be written as 
\begin{eqnarray}
\psi(u,t)&=&\sum_{k\in\mathcal{I}}\P(\tau_u = t_k)\nonumber\; =\;\sum_{k\in\mathcal{I}}\frac{t_0}{t_k}\;\mathbb{P}\left[N(t_k) = k\right],
\end{eqnarray}
where the last equality follows from applying Corollary 3.4 of Goffard and Lefevre \cite{GoLe17JMAA} and is equivalent to \eqref{eq:finite_time_ruin_proba}. 
\item This result is standard, see e.g.\ \cite[Th.VI.2.1]{AsAl10}. For the sake of completeness, we still prefer to give the main idea behind its derivation and the probabilistic reason for Equation \eqref{lueq} here. The process $\{e^{\theta (ct-b\tilde{N}_t) -t\kappa(\theta)}\}_{t\geq 0},$ is a martingale for any $\theta\in{\mathbb R}$, where $\kappa(\theta) = \log\E(e^{c\theta - b\theta  \tilde{N}_1})$, see e.g. \cite[Th.II.2.1]{AsAl10}. 
%
%
The simplest choice for $\theta$ is the one for which $\kappa(\theta)=0$ , i.e.
$$
c\theta + p\lambda (e^{-b\theta}-1)=0. 
$$
The latter equation admits a unique non-negative solution $\theta^\ast$, and the process $\left(e^{\theta^\ast (ct-b\tilde{N}_t)}\right)_{t\geq0}$ is therefore a martingale. It then  follows by the Optional Stopping Theorem (cf.\ \cite[Prop.II.3.1]{AsAl10}) that 
\begin{equation*}
\psi(u) =e^{-\theta^\ast u},
\end{equation*}
since under the assumptions on $R_t$ (with only upward jumps) the surplus at the time of ruin is necessarily zero. 
\item Using the tower property, we can express the value function \eqref{eq:value_function_finite_time} as 
\begin{eqnarray}
V(u,t) &=&  \E\left[\E\left(R_t\mathbb{I}_{\tau_u>t}\Big\rvert N_t\right)\right]\nonumber\\
&=&\E\left[\left(u+bN_t - ct\right)\E\left(\mathbb{I}_{\tau_u>t}\Big\rvert N_t\right)\right]\nonumber\\
&=&\E\left[\left(u+bN_t - ct\right)_+\P\left(\bigcap_{k = 1}^{N_t}\{T_k < t_{k-1}\land t\}\Big\rvert N_t\right)\right]\nonumber\\
&=&\E\left[\left(u+bN_t - ct\right)_+\P\left(\bigcap_{k = 1}^{N_t}\{U_{1:k} < \frac{t_{k-1}}{t}\land 1\}\right)\right],
\end{eqnarray}
where $U_{1:n},\ldots, U_{n:n}$ denote the order statistics of $n$ \iid standard uniform random variables. Using the interpretation of Abel-Gontcharov polynomials as the joint probabilities of uniform order statistics, see Goffard and Lefevre \cite[Prop.2.4]{GoLe17JMAA} then yields \eqref{eq:value_function_finite_time_prop}.
\end{enumerate}
\end{proof}
\noindent Expression \eqref{eq:value_function_finite_time_prop} is unfortunately not of closed form. The Abel-Gontcharov polynomials may be evaluated recursively, and acceptable approximations follow from simple truncation or simulation. However, we are interested in a more explicit and amenable expression for this quantity. In addition, the choice of the considered time horizon $t$ is somewhat arbitrary. We therefore propose now to replace the fixed time horizon $t$ by an independent exponential random time horizon $T$ with mean $t$. This will allow us to derive a simple and explicit expression for the respective quantity, which hopefully will provide a good approximation of $V(u,t)$. At the same time, this random time horizon also has an intuitive interpretation on its own: it can be seen as the first epoch of an independent homogeneous Poisson process, so that there is a lack-of-memory as to when exactly the surplus level is going to be evaluated, with the expected value of that time horizon being $t$. We hence consider the value function
\begin{equation}\label{eq:value_function_exp_time_horizon}
\vh(u,t):= \E(V(u,T)) = \E(R_T\mathbb{I}_{\tau_u>T}),
\end{equation} 
where $T\sim \text{Exp}(t)$. 
\begin{prop}\label{prop:value_function_protocol}
For any $u\ge 0$, the value function $\vh(u,t)$ defined in \eqref{eq:value_function_exp_time_horizon} satisfies
\begin{equation}\label{eq:value_function_exp_time_horizon_formula}
\vh(u,t) = u+(p\lambda b - c)\,t\,(1 - e^{\rho^{\ast}u}),
\end{equation}
where $\rho^{\ast}$ is the negative solution of the equation
\begin{equation}\label{eq:equation_rho}
-c\rho +p\lambda \,(e^{b\,\rho }-1) = 1/t.
\end{equation}
\end{prop}
\begin{proof}
Let $0<h<u/c$, so that ruin can not occur in the interval $(0,h)$. We distinguish three cases:
\begin{itemize}
	\item[(i)] $T>h$ and there is no block discovery in the interval $(0,h)$,
	\item[(ii)] $T<h$ and there is no block discovery in the interval $(0,T)$,
	\item[(iii)] There is a block discovery before time $T$ and in the interval $(0,h)$. 
\end{itemize}
By conditioning we see that 
\begin{eqnarray*}
	\vh(u,t)& =&e^{-h(1/t + p\lambda)}\,\vh(u-ch,t)+\int\limits_0^h\frac1t\, e^{-s(1/t + p\lambda)}\,(u-cs)ds\\
	&&\quad+\int\limits_0^h p\lambda\, e^{-s(1/t + p\lambda)}\,\vh(u-cs+b,t)ds.
	\end{eqnarray*}
Now we take the derivative with respect to $h$ and set $h=0$ to obtain

\begin{equation}\label{eq:ODE}
c\vh'(u,t) + \left(\frac{1}{t} +  p\lambda\right)\vh(u,t) - p\lambda \vh(u+b,t) - \frac{u}{t} =0,
\end{equation}
with boundary condition $\vh(0,t) = 0$ and such that $0\leq \vh(u,t)\leq u-ct+p\lambda t$ for $u>0$. 
 
Here the derivative is with respect to the first argument (note that the above derivation automatically guarantees the existence of the derivative $\vh'(u,t)$). 	
Equation \eqref{eq:ODE} is a particular case of an advanced functional differential equation (concretely, a differential equation with an advanced argument). By construction, it will have a solution (namely our value function), although we would like to remark here that, from a mathematical point of view, to the best of our knowledge there is no result available that guarantees the existence and uniqueness of a solution to such an equation  
in general
(typically, one would expect the knowledge of $\vh(u,t)$ on the entire interval  $u\in[0,b]$ as an initial condition). For instance, the form does not even fall into the rather general setup studied by Schauder fixed point techniques in \cite{banas} (as the asymptotics assumptions (iii) and (iv) of \cite[Sec.3]{banas} do not hold here), and the argument $u$ is not restricted to a compact interval which could accommodate the problem into the framework studied in \cite{jankowski}, see also \cite{iakovleva} and \cite[Ch.3]{smith} for further details. However, the concrete form of \eqref{eq:ODE} will allow to derive its solution, and one can re-confirm its uniqueness by stochastic simulation of $\vh(u,t)$, which then suffices for our purpose. We are seeking a solution of the form 
\begin{equation}\label{eq:potential_solution}
\vh(u,t) = Ae^{\rho u }+Bu + C,\text{ }u \ge 0, 
\end{equation}
where $A, B,C$ and $\rho$ are constants to be determined. Substituting \eqref{eq:potential_solution} in \eqref{eq:ODE} together with the boundary condition yields the system of equations 
\begin{equation*}
\begin{cases}
0&=ct\rho + \left(1+p\lambda t\right)-p\lambda te^{\rho b}, \\
0&= B\left(1+tp\lambda\right)-p\lambda tB - 1,\\
0&=Bct+C(1+tp\lambda) - p\lambda t Bb-p\lambda tC, \\
0&=A+C.
\end{cases}
\end{equation*}
We then have $A = -t(p\lambda b - c)$, $B = 1$, $C = t(p\lambda b - c)$ and $\rho$ is solution of Equation \eqref{eq:equation_rho},
which admits one negative and one positive solution. As $A<0$, we have to choose the negative solution $\rho^\ast<0$ in order to ensure $\vh(u,t)>0$. Substituting $A,B,C$ and $\rho^{\ast}$ in \eqref{eq:potential_solution}  yields the result.
\end{proof}
\begin{rem}\label{rem21}\normalfont 
Note that, as the solution of \eqref{eq:equation_rho}, $\rho^*$ can be expressed through the Lambert W function, see Corless et al. \cite{corless1996lambertw}, with 

	\begin{equation*}
	\rho^{\ast}=-\frac{p \lambda t+1}{ct}
	-\frac{1}{b} \,{\rm W} \left(-\frac{p\lambda
		\,b}{c}\,{e^{-b\,\left(\frac{p \lambda t+1}{ct}\right)}}
	\right).
	\end{equation*} 
	Furthermore, the similarity of the equations \eqref{lueq} and \eqref{eq:equation_rho} (referred to as \textit{Lundberg equations} in risk theory) is no coincidence. In contrast to the probabilistic derivation of \eqref{lueq} presented in the proof of Proposition  \ref{prop:ruin_proba_and_value_func}, one could also mimick the analytic approach above and identify $\psi(u)$ as the solution of the differential equation $c\psi'(u)+p \lambda(\psi(u)-\psi(u+b))=0$, of which \eqref{lueq} is the characteristic equation (here $\theta=-\rho$).\hfill $\Box$
\end{rem}

\noindent We can now complement the results of Proposition \ref{prop:ruin_proba_and_value_func}. In fact, for an exponential time horizon the ruin probability has a very simple form (see also Asmussen et al.\ \cite{asmussen2002} for ruin probabilities up to a random time horizon in the classical risk model).  

\begin{coro}\label{next} 
	\text{ }
	 For $u\ge 0, t>0$, let $\ph(u,t):=\E(\psi(u,T))=\P(\tau_u<T)$ denote the probability that ruin takes place before time $T$, where  
	$T$ is an independent exponential random variable with mean $t$. Then 
	\begin{equation}\label{nextfinite_time_ruin_proba}
	\ph(u,t) =e^{\rho^\ast u},
	\end{equation}
	where $\rho^\ast$ is the negative solution in $\rho$ of Equation \eqref{eq:equation_rho}.
\end{coro}
\begin{proof}
	Following up on Remark \ref{rem21}, using conditioning the ruin probability $\ph(u,t)$ can be seen to be the solution of the differential equation
	$$c\ph'(u,t)+(p \lambda+1/t)\,\ph(u,t)-p \lambda\,\ph(u+b,t)=0$$
	with initial condition $\psi(0,t)=1$ and boundary condition $\lim_{u\to\infty}\ph(u,t)=0$, where the derivative is with respect to the first argument. This is in fact the homogeneous equation of \eqref{eq:ODE} and gives a nice interpretation of the term $1/t$ as the intensity of the dynamic time horizon when moving along the surplus trajectory $R_t$. Trying a solution of the form $\ph(u,t)=Ae^{\rho u}+C$ then leads immediately to the result, since the boundary condition forces $C=0$ and the initial condition gives $A=1$, whereas the exponent $\rho$ has to be the negative solution of \eqref{eq:equation_rho}.  
\end{proof}
\noindent 
\begin{rem}\normalfont
Note that 	
\begin{eqnarray*}\ph(u,t)&=&\P(\tau_u<T)\\
	&=&\E(\P(T>\tau_u\vert \tau_u))\\
	&=&\E(e^{-\delta\tau_u})
\end{eqnarray*}	
with $\delta=1/t$. Thus, $\ph(u,t)$ is the Laplace transform of the ruin time density. In actuarial language, it is the expected present value (evaluated with the force of interest $\delta$) of 1 that is paid at the time of ruin. 	\\
It can be verified that the process $\{\exp( - \delta t + \rho^* R_t)\}_{t\ge 0}$  is a martingale. From the Optional Stopping Theorem it follows that the expectation at the time of ruin is equal to its initial value, which is precisely Formula \eqref{nextfinite_time_ruin_proba}, constituting another proof of that formula, in the spirit of the proof of \eqref{eq:infinite_time_ruin_proba} given above.
%
\end{rem}

Comparing the expression \eqref{eq:finite_time_ruin_proba} for $\psi(u,t)$ with deterministic time horizon $t$ and  \eqref{nextfinite_time_ruin_proba}  
for $\ph(u,t)$ with random time horizon (of the same expected value) shows the substantial simplification one can achieve through the principle of randomization. We refer to Section \ref{sec:numerical_illustrations} for a numerical comparison of the two quantities.

\section{Ruin and expected surplus of a selfish miner}\label{sec:blockwithholding_strategy}
It has been discussed in the literature for a while already that it may be advantageous to keep newly found blocks hidden (see e.g.\ \cite{EySi18}). Such a block withholding strategy is referred to as selfish mining. The goal is to build up a stock of blocks and release them publicly at well chosen times so as to fork the chain and make a part of the mining activities of competitors worthless, depending on which branch is followed up in the longer run. A fork happens when two equally long versions of the blockchain coexist, and it resolves whenever one of them becomes longer by at least one block. In the sequel, we consider a variant of the selfish mining strategy introduced by Eyal and Sirer \cite{EySi18}. Let again $(N_t)_{t\geq0}$ be the homogeneous Poisson process with intensity $\lambda$ that governs the block discovery process of the entire network. As in the previous section, we consider a miner named Sam who owns a share $p\in(0,1)$ of the computing power so that a newly found block belongs to Sam with probability $p$. We keep track of the number of blocks withheld by Sam via a Markov jump process
$$
X_t=Z_{N_t}, \text{ }t\geq0,
$$ 
where $(Z_{k})_{k \geq 0}$ is a homogeneous Markov chain with finite state space $E = \{0,1,0^\ast\}$ which we define now. Assume that at time $t\geq0$, the miners have published $N_t = k$ blocks.
\begin{itemize}
\item If Sam is not hiding any block, then $Z_k = 0$, 
\begin{itemize}
\item if Sam finds the next block, he stores it in a buffer and $Z_{k+1} = 1$, 
\item if the other miners discover a block then Sam's buffer remains empty $Z_{k+1}=0$, 
\end{itemize}
In both cases, Sam is not collecting any reward.
\item If Sam is hiding one block at that time, then $Z_k = 1$,
\begin{itemize} 
\item if he then finds a new block, then he broadcasts both blocks immediately (which resets the Markov chain to $Z_{k+1}=0$), and he collects two rewards,
\item if the others find a block, then Sam also releases his block leading to a fork situation characterized by $Z_{k+1} = 0^\ast$. At that moment Sam is not collecting any rewards.
\end{itemize}
\item If a fork situation is present at that time ($Z_k = 0^{\ast}$), then 
\begin{itemize}
\item if Sam finds a new block then he appends it to his branch of the chain and collects the reward for two blocks and $Z_{k+1} = 0$.  
\item if the others find a block then 
\begin{itemize}
\item they append it to Sam's branch with a probability $0\leq q \leq 1$, in which case Sam gets the reward for one block. 
\item If the block is mined on top of the competing branch, then Sam earns nothing. 
\end{itemize}
In both cases, the number of hidden blocks then becomes $Z_{k+1}=0$.
\end{itemize}
\end{itemize}
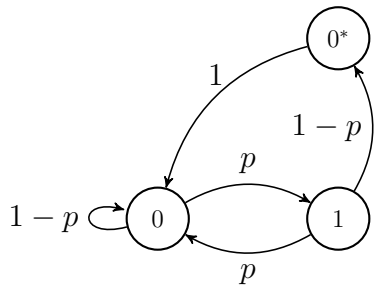
\begin{figure}[h]
\begin{center}
\begin{tikzpicture}[->, >=stealth', auto, semithick, node distance=3cm]
\tikzstyle{every state}=[fill=white,draw=black,thick,text=black,scale=0.8]
\node[state]    (1)                     {$0$};
\node[state]    (2)[right of=1]         {$1$};
\node[state]    (3)[above of=2]         {$0^{\ast}$};
\path
(1) edge[loop left]     node{$1-p$}        (1)
    edge[bend left]     node{$p$}          (2)
(2) edge[bend left]      node{$p$}         (1)
    edge[bend right]      node{$1-p$}      (3)
(3) edge[bend right, above]      node{$1$}         (1);
\end{tikzpicture}
\end{center}
\caption{Transition graph of the Markov chain $(Z_k)_{k\geq0}$ representing the stock of blocks retained by Sam when implementing the simplified selfish mining strategy.}
\label{fig:transition_graph}
\end{figure}
Figure \ref{fig:transition_graph} depicts the transition graph of $(Z_{k})_{k\geq0}$. The selfish mining strategy alters the reward collecting process. The surplus process of Sam introduced in \eqref{eq:surplus_protocol} now becomes 
\begin{equation}\label{eq:surplus_selfish}
R_t = u-c\cdot t+b\cdot\sum_{n = 1}^{N(t)}f\left[Z_{n-1},\xi_n, \zeta_n\right],
\end{equation}
where the $(\xi_n)_{n\ge 1}$ and $(\zeta_n)_{n\ge 1}$ are \iid Bernoulli random variables with parameter $p$ and $q$, respectively, and 
\begin{equation}\label{eq:gain_function_selfish}
f\left[Z_{n-1}, \xi_n,\zeta_{n}\right] = \begin{cases}
0,&\text{ if } Z_{n-1} =0,\\
0,&\text{ if } Z_{n-1} =1\text{ }\&\text{ }\xi_{n} =0,\\
2,&\text{ if } Z_{n-1} =1\text{ }\&\text{ }\xi_{n} =1,\\
0,&\text{ if } Z_{n-1} =0^\ast \text{ }\&\text{ }\xi_{n} = 0\text{ }\&\text{ }\zeta_{n} = 0,\\
1,&\text{ if } Z_{n-1} =0^\ast \text{ }\&\text{ }\xi_{n} = 0\text{ }\&\text{ }\zeta_{n} = 1,\\
2,&\text{ if }Z_{n-1} = 0^\ast \text{ }\&\text{ }\xi_{n} = 1.\\
\end{cases}
\end{equation} 
\begin{rem}\normalfont 
In contrast to the selfish mining strategy defined above, Eyal and Sirer's \cite{EySi18} more general selfish mining strategy does not automatically release the blocks when the buffer reaches size two. On the contrary, the stock is being built up continuously and the selfish miner releases one block for every block found by the others. As soon as the level of the buffer goes back to two, the selfish miner releases the last two blocks and collects the reward for all the blocks he has been hiding. The process $(Z_n)_{n\geq0}$ remains a proper Markov chain, and the study of its stationary distribution allowed Eyal and Sirer \cite{EySi18} to show how the relative revenue of the selfish miner exceeds his fair share. However, in Eyal and Sirer's setting the reward collecting process and the surplus process are no longer Markovian, rendering an explicit expression for the value function infeasible. In addition, we believe that postponing the capital gain for too long would impair the miner's solvency to the point of making this strategy unsustainable from a financial point of view. Consequently, for the consideration of solvency aspects we prefer to focus on the simpler selfish mining strategy here. \hfill $\Box$
\end{rem}
\noindent It is interesting to see whether selfish mining is still profitable for Sam if the possibility of ruin is included in the analysis. His average earning per time unit is now given by 
\begin{equation}\label{eq:average_earning_selfish}
\gamma = \frac{b}{t}\,\E\left[\sum_{k = 1}^{N(t)}f(Z_{k-1},\xi_k,\zeta_k)\right] - c.
\end{equation}
This quantity can be determined if we assume that Sam has been mining in a selfish way for quite some time already, so that we can consider the Markov chain to be in stationarity with stationary probabilities
$$
\P(Z = 0)=\frac{1}{1+2p-p^2},\text{ }\P(Z = 1)=\frac{p}{1+2p-p^2},\text{ and }\P(Z = 0^{\ast})=\frac{p(1-p)}{1+2p-p^2}.
$$
The quantities $U_k:=f\left(Z_{k-1},\xi_k,\zeta_k\right),\text{ }n\geq1$ then have a \pmf $p_U(\cdot): =\mathbb{P}(U = \cdot)$ given by 
$$
p_U(0) =\frac{1+p(1-p)+p(1-p)^2(1-q)}{1+2p-p^2},\text{ }p_U(1)=\frac{pq(1-p)^2}{1+2p-p^2},\text{ and }p_U(2)=\frac{p^2+p^2(1-p)}{1+2p-p^2},
$$
and the net profit condition correspondingly reads
\begin{equation}\label{netpro}
\gamma = b\lambda\frac{qp(1-p)^2 + 4p^2-2p^3}{1+2p-p^2} - c>0.
\end{equation}
The profitability of selfish mining consequently depends on the interplay between the probabilities $p$ and $q$, and not all values of $p$ and $q$ will lead to positive expected profit. If we assume for a moment that the $U_k$'s are \iid with \pmf $p_U$ (which is of course a strongly simplifying assumption) and the net profit condition holds, then the infinite-time ruin probability is given by
$$
\psi(u) = e^{-\theta^\ast u},
$$
where $\theta^\ast$ is the unique positive root of the equation 
$$
{c\theta}+\lambda\left[p_U(0)+e^{-b \theta}p_U(1)+e^{-2b \theta}p_U(2)-1\right]=0.
$$
The proof is anologous to the one of Proposition \ref{prop:ruin_proba_and_value_func}. Deriving the actual ruin probability and the expected surplus under ruin constraints at some deterministic time horizon is more difficult than in the previous section. However, for the exponential time horizon (with mean $t$), we are able to derive an explicit formula for $\vh(u,t)$. Since the reward process is modulated by the Markov chain, we have to derive the value function depending on the current state of the Markov chain. The result below in principle provides a formula for
\begin{equation}\label{eq:value_function_selfish}
\vh_z(u,t)\equiv \E(V_z(u,T)) = \mathbb{E}\left(R_T\mathbb{I}_{\tau_u>T}\Big \rvert Z_0 = z\right)
\end{equation}
with $T\sim\text{Exp}(t)$ for any $z\in\{0,1, 0^{\ast}\}$. The most interesting among these is $\vh_0(u,t)$, which refers to the situation where the buffer is empty at the beginning, and we formulate the result in terms of the latter. 
\begin{theo}\label{theo:value_Function_Selfish}
For any $u\ge 0$, the value function under selfish mining as defined above is given by
\begin{equation}
\vh_0(u,t)={\it A_1}\,{{\rm e}^{\rho_1\,u}}+{{\rm e}^{\rho_2\,u}} \left[ {\it A_2}\,
\cos \left( \rho_3\,u \right) +{\it A_3}\,\sin \left( \rho_3\,u \right)\right] +u+C,\label{thisone}
\end{equation}
where $\rho_1$ and $\rho_2\pm i\rho_3$ are the only solutions in $\rho$ of the equation
\begin{equation}\label{solu}
\left( {\it D_1}\,\rho+{\it D_2} \right) {{\rm e}^{2\,\rho\,b}}+{\it D_3
}\,{{\rm e}^{\rho\,b}}= {\it D_4}\,{\rho}^{3}+{\it D_5}\,{\rho}^{2}+{\it 
D_6}\,\rho+{\it D_7}
\end{equation}
with negative real parts. The constants $C, D_1,\ldots, D_7$ are given by 
\begin{eqnarray}
C &=& -\frac{\lambda^{2} {t}^{3} \left\{\lambda pb \left[  \left( q-2 \right) {p
	}^{2}+ \left( 4-2\,q \right) p+q \right] +c \left( {p}^{2}-2\,
	p-1 \right)  \right\}}{{\lambda}^{2}{t}^{2} \left( {
		p}^{2}-2\,p-1 \right) - \left( p+2 \right) \lambda\,t-1}\nonumber\\
&&-\frac { \lambda\,t^2\, \left( 2b{p}^{2}\lambda-
	c\left( p+2 \right)  \right)-ct }{{\lambda}^{2}{t}^{2} \left( {
		p}^{2}-2\,p-1 \right) - \left( p+2 \right) \lambda\,t-1}\label{cconst}\\
D_1 &=& {\frac {pc}{1-p}}\text{, }D_2\hspace{0.2cm}=\hspace{0.2cm} {\frac {p \left( 1+t \left( 2-p \right) \lambda \right) }{t \left( 1-
p \right) }}\text{, }D_3\hspace{0.2cm} =\hspace{0.2cm} \left( 1-p \right) \lambda\,q,\text{ }D_4\hspace{0.2cm} =\hspace{0.2cm}{\frac {{c}^{3}}{{\lambda}^{2} \left( 1-p \right) p}},\nonumber\\
D_5 &=& {\frac { {c}^{2}\left( 3+\lambda\,t \left( p+2 \right)  \right) }{{
 \lambda}^{2}t \left( 1-p \right) p}},\text{ }D_6\hspace{0.2cm} =\hspace{0.2cm}{\frac { \left( 3+\lambda\,t \left( 2p+1 \right)  \right) c
 \left( \lambda\,t+1 \right) }{{t}^{2}{\lambda}^{2} \left( 1-p\right) p}}\text{, and }\nonumber\\
D_7 &=& {\frac {1+\lambda\,t \left( p+2 \right)+ {
			\lambda}^{2}{t}^{2}\left( 2p+1 \right)+{\lambda}^{3}{t}^{3}p \left( p(1-q) ( 2-p)+q \right)  }{{\lambda}^{2}{t}^{3}
 \left( 1-p \right) p}}.\nonumber
\end{eqnarray}
The constants $A_1, A_2$ and $A_3$ are the solution of the linear equation system 
\begin{equation}\label{system}\left(\begin{array}{ccc}
1 & 1 &0 \\
\rho_1 & \rho_2 & \rho_3 \\
B_1 &B_2& B_3
\end{array}\right)\left(\begin{array}{c}
A_1 \\A_2 \\A_3 
\end{array}\right)=\left(\begin{array}{c}
-C \\-1 \\-2b-C -\frac{1}{ct}
\end{array}\right),\end{equation}
with \begin{eqnarray}
	B_1&=&-\rho_1^2+\lambda^2p^2\,e^{2b \rho_1}/c^2,\nonumber\\
	B_2&=&-\rho_2^2+\rho_3^2+\lambda^2p^2\,e^{2b \rho_2}\cos(2b \rho_3)/c^2,\label{biis}\\
		B_3&=&-2\rho_2\rho_3+\lambda^2p^2\,e^{2b \rho_2}\sin(2b \rho_3)/c^2.\nonumber
\end{eqnarray}
\end{theo}
\noindent The proof is delegated to Appendix \ref{Appendix}. 

\begin{rem}\label{rem32}\normalfont 
If the connectivity parameter $q$ equals $1$, then the selfish miner will not waste any blocks by withholding them first, he will just append them later. In that sense his surplus process, and correspondingly the value function, is very similar to the one when following the protocol, with the subtle difference that by receiving the rewards later, the likelihood of ruin is increased (reducing the value function). When the initial capital increases indefinitely, that effect evaporates and then the only remaining difference between the two is that for the expected surplus at the evaluation time horizon $T\sim Exp(t)$ the initial state still matters for the selfish miner, since starting -- with the same initial capital $u$ -- in state $0^*$ or $1$ will lead to overall one more block than when starting in state $0$ (i.e., with an empty buffer). This latter difference should then itself disappear when $t$ increases indefinitely. Checking this, $\vh(u,t)-u$ in \eqref{eq:value_function_exp_time_horizon_formula} converges to $(p\lambda b-c) t$, and $\vh_0(u,t)-u$ in \eqref{thisone} converges to $C$, as $u\to\infty$. While $C$ defined in \eqref{cconst} is a rather involved function of the parameter $t$, a simple calculation shows that, for $q=1$ (and only then!), $C/t\to p\lambda b-c$ as $t\to\infty$, so that the two expressions are indeed asymptotically equivalent. For a numerical comparison of $\vh_0(u,t)$ and $\vh(u,t)$ for finite values of $u$ and $t$, see Section \ref{sec:numerical_illustrations}. \hfill$\Box$
\end{rem}

\noindent Finally, the ruin probability for the selfish miner up to an exponential time horizon can be obtained in an analogous (and slightly simpler) way. 	
\begin{coro}\label{coro2}
	For any $u\ge 0$, the ruin probability $\ph_0(u,t)$ up to an exponential time horizon $T\sim \text{Exp}(t)$ for a selfish miner starting in state $0$ is given by
\begin{equation}\label{ruin_selfish}
\ph_0(u,t) =
{\it C_1}\,{{\rm e}^{\rho_1\,u}}+{{\rm e}^{\rho_2\,u}} \left[ {\it C_2}\,
\cos \left( \rho_3\,u \right) +{\it C_3}\,\sin \left( \rho_3\,u \right)\right],
\end{equation}
where $\rho_1$ and $\rho_2\pm i\rho_3$ are (again the same) only solutions in $\rho$ of Equation \eqref{solu} with negative real parts. 
The constants $C_1, C_2$ and $C_3$ are the solution of the linear equation system 
\begin{equation}\label{systemruin}\left(\begin{array}{ccc}
1 & 1 &0 \\
\rho_1 & \rho_2 & \rho_3 \\
B_1 &B_2& B_3
\end{array}\right)\left(\begin{array}{c}
C_1 \\C_2 \\C_3 
\end{array}\right)=\left(\begin{array}{c}
1 \\-\frac{1}{ct}\\\frac{\lambda^2p^2}{c^2}-\frac{1}{c^2t^2}
\end{array}\right),\end{equation}
where $B_1,B_2,B_3$ are again given by \eqref{biis}.
\end{coro}
\noindent The proof is delegated to Appendix \ref{AppendixB}.

\section{Numerical illustrations and sensitivity analysis}\label{sec:numerical_illustrations}
In this section, we are going to use the formulas for the ruin probabilities and expected surplus in case ruin did not occur derived in the previous sections for numerical illustrations and in particular also for sensitivity analysis with respect to the involved parameters. We will first study the profitability of following the protocol in Section \ref{sub:numerical_illustration_pro} and proceed to the profitability of selfish mining in Section \ref{sub:numerical_illustration_self}. We define the expected profit as the expected surplus under ruin constraints minus the initial capital.\\

\noindent Throughout this section, the time unit is one hour. Since the Bitcoin blockchain protocol is designed to ensure that one block of confirmed transactions is added to the blockchain about every ten minutes, this renders the block arrival intensity in our model to be $\lambda = 6$. The reward $b$ is determined by the number $n_{BTC}$ of bitcoins earned when finding a block and the price $\pi_{BTC}$ of the bitcoin. For the illustrations in this paper, we use the data of January 1, $2020$, when $n_{BTC}=12.5$ and $\pi_{BTC} =\$7,174.74$\footnote{Source:\href{https://www.blockchain.com/}{blockchain.com}}, so that the reward amounts to
$$
b = n_{BTC}\times \pi_{BTC} = \$89,684.30.
$$
We assume that the operational cost of mining reduces to the electricity consumed when computing hashes. On January 1,  $2020$, the yearly consumption of the network was estimated by the Cambridge Bitcoin Electricity Consumption Index\footnote{Source:\href{https://www.cbeci.org/}{CBEI}} to $72.1671$ TWh.\footnote{The choice of this concrete date for the illustrations in this paper is somewhat arbitrary. Choosing another date and estimate of the Bitcoin price will, however, not crucially change the conclusions as long as the reward for finding a block compensates the operational cost.    
} We denote by 
$$
W = \frac{72.1671\times 10^9}{365.25\times 24}
$$
the electricity consumption of the network expressed in kWh. We let the operational cost $c$ of a given miner be proportional to its share $p\in(0,1)$ of the network computing power with 
$$
c = p\times W \times \pi_W,
$$  
where $\pi_W$ denotes the price of the electricity where the miner is located, expressed in USD per kWh. 

\subsection{Expected profit computations for a miner following the protocol}\label{sub:numerical_illustration_pro}
We assume now the assumptions of risk model \eqref{eq:surplus_protocol} for the wealth of a miner that follows the protocol. The net profit condition $p\lambda b-c\geq0$ then translates, in terms of the price of electricity, to 
$$
\pi_W<\frac{\lambda b}{W} = 0.065.
$$
In the sequel, unless otherwise stated, we set $\pi_W = 0.06$ (so the net profit condition is satisfied) and assume the miner to have a share 
$p = 0.1$ of the network hashpower. Figure \ref{sub:pro_rp_MC} displays the infinite-time and finite-time ruin probability (with chosen time horizon to be 6 hours) as a function of initial wealth using the formulas derived in Proposition \ref{prop:ruin_proba_and_value_func}. \\
\begin{figure}[h!]
	\begin{center}
		\subfloat[$\psi(u)$ (solid), $\psi(u, 6h)$ (dashed) and $\ph(u,6h)$ (dotted) as a function of initial wealth $u$.]{
			\includegraphics[width=0.45\textwidth]{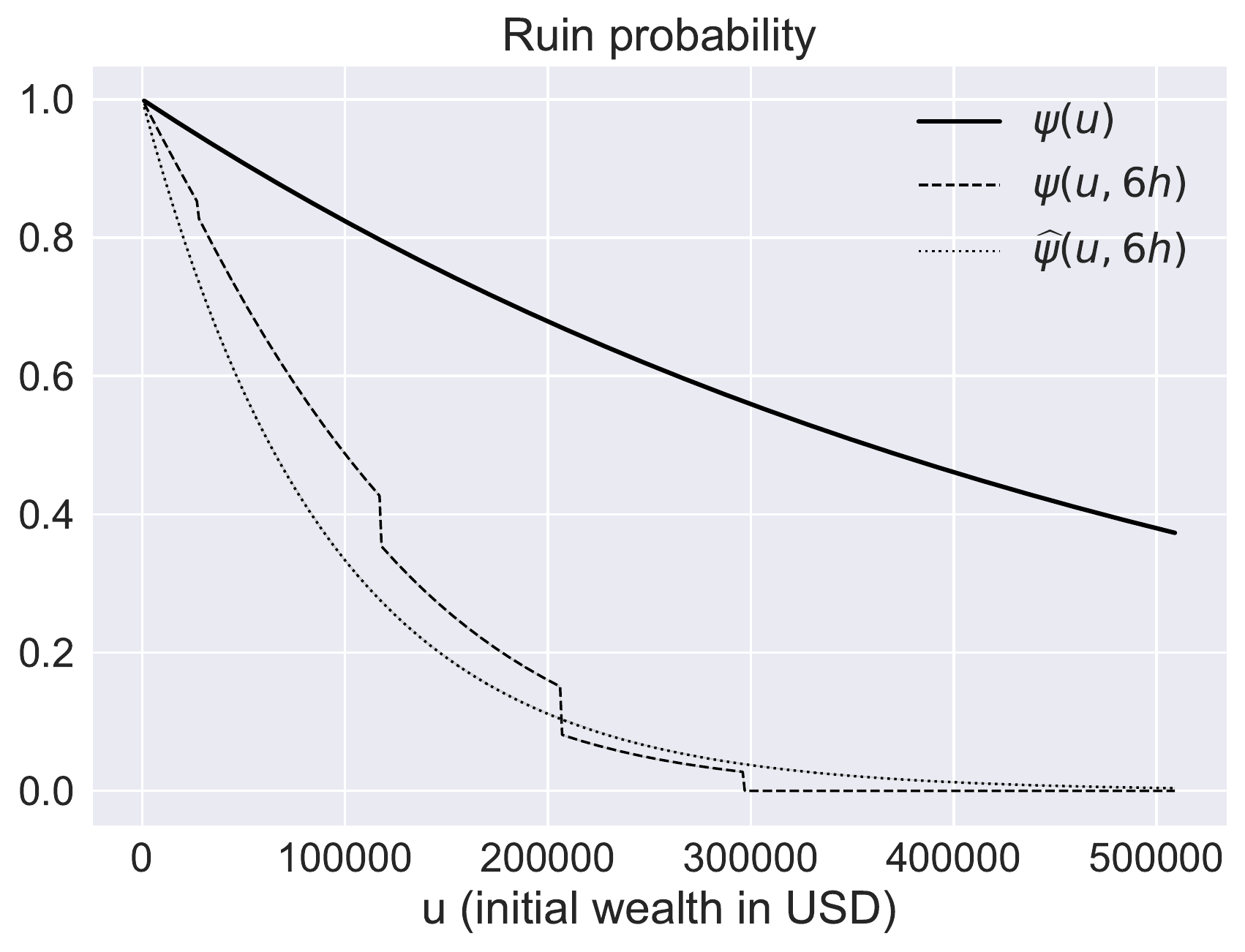}
			\label{sub:pro_rp_MC}
		}
		\hskip1em
		\subfloat[$\mathbb{E}(R_{6h})-u$ (solid) , $V(u,6h)-u$ (dashed), $\vh(u,6h)-u$ (dotted) and as a function of initial wealth $u$.]{
			\includegraphics[width=0.45\textwidth]{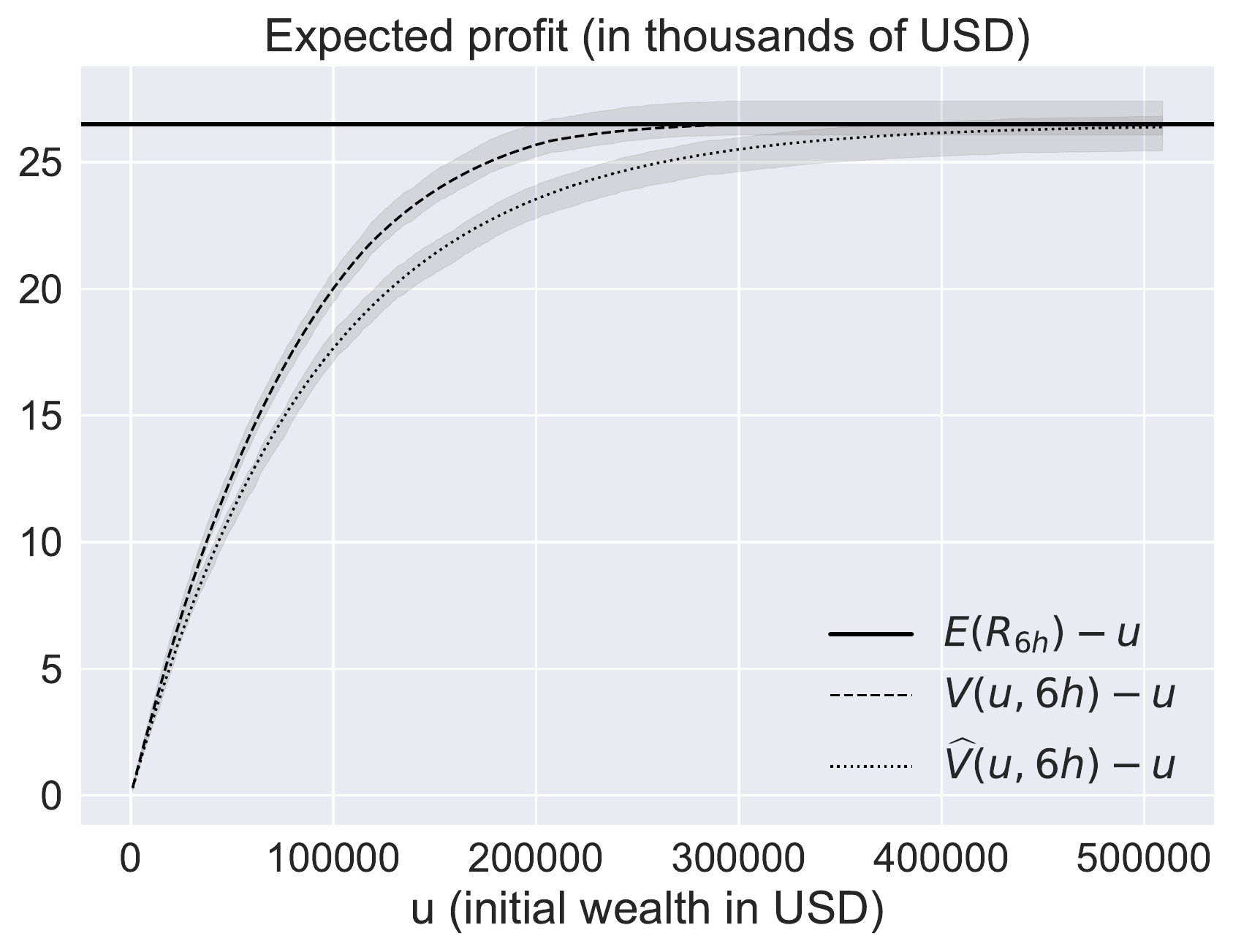}
			\label{sub:pro_rev_MC}
		}
		\caption{Ruin probabilities and expected profit as a function of initial wealth for a miner that follows the protocol with hashpower $p=0.1$ and electricity price $\pi_W=0.06$, together with the confidence bands of a Monte-Carlo simulation.}
		\label{fig:pro_ruin_rev_MC}
	\end{center}
\end{figure}

\noindent Figure \ref{sub:pro_rev_MC} plots the expected profit at that deterministic time horizon $t=6$ hours (recall that the expected profit corresponds to the expected surplus under ruin constraints net of the initial wealth $u$) using formula \eqref{eq:value_function_finite_time_prop} and compares it to $\mathbb{E}(R_t)-u = (p\lambda b -c)t$, which is the expected surplus, again net of the initial wealth, of a miner without that ruin constraint (we will refer to that latter quantity as the \textit{target profitability} in the sequel). One sees that with increasing initial capital the two quantities get closer, as the event of ruin then becomes more and more unlikely. In Figure \ref{sub:pro_rev_MC} we also give the quantity $\vh(u,t)$ as derived in Proposition \ref{prop:value_function_protocol} for an exponential time horizon $T$ with expected value 6 hours, net of initial capital $u$ (dashed line). When comparing the deterministic time horizon to the random time horizon with the same expected value, it is clear that the exponential distribution overweights time horizons below the expected value of 6 hours, so that it is intuitive that the expected gain is below the one for the deterministic time horizon (as, due to the
net profit condition, the expected income increases with the time horizon). But this difference diminishes for large values of $u$, as they both approach the target profitability. In order to double-check the validity of the obtained explicit formulas, we also plot in Figure \ref{fig:pro_ruin_rev_MC} asymptotic 95\%-confidence bands of an independent Monte Carlo simulation of 
all the quantities using 250'000 sample paths, and the exact formulas are nicely in the center of these bands (in Figure \ref{sub:pro_rp_MC} these bands are so narrow that one can barely observe them with the naked eye). Another way to read Figure \ref{sub:pro_rp_MC} is that it quantifies the amount of initial capital $u$ needed in order to ensure survival within 6 hours, or at any time in the future, for any given probability. Note that as $u$ increases, more and more terms in the finite-sum expression \eqref{eq:finite_time_ruin_proba} for $\psi(u,t)$ disappear, causing the jumps in the curve. Figure \ref{sub:pro_rp_MC} also suggests that switching from a deterministic to a random exponential time horizon with the same mean does not impact the resulting ruin probability substantially. This further justifies the switch from a deterministic to a random time horizon with its enormous computational advantages, particularly for the selfish mining case. The time horizon of 6 hours in Figure \ref{fig:pro_ruin_rev_MC} was in part chosen to ensure a satisfactory accuracy when evaluating formula \eqref{eq:value_function_finite_time_prop}. 
While there is no issue for $\psi(u)$ and $\psi(u,t)$, the recursive computation of Abel-Gontcharov polynomials in that formula for $V(u,t)$ becomes numerically unstable for larger time horizons.  \\

\noindent Figure \ref{fig:pro_ruin_rev_u} displays the ruin probability and expected profit for the larger time horizons 1 day, 1 week and 2 weeks for the otherwise same set of parameters. Figure \ref{sub:pro_rp_u}  illustrates quantitatively how  $\ph(u,t)$ approaches the infinite-time ruin probability $\psi(u)$ as the expected time horizon $t$ increases.\\

\begin{figure}[h!]
  \begin{center}
    \subfloat[$\ph(u,t)$ as a function of initial wealth $u$ for varying time horizon: (dotted) $t=1d$, (dashed) $t=1w$, (dash-dotted) $t = 2w$, and (solid) $t=\infty$.]{
      \includegraphics[width=0.45\textwidth]{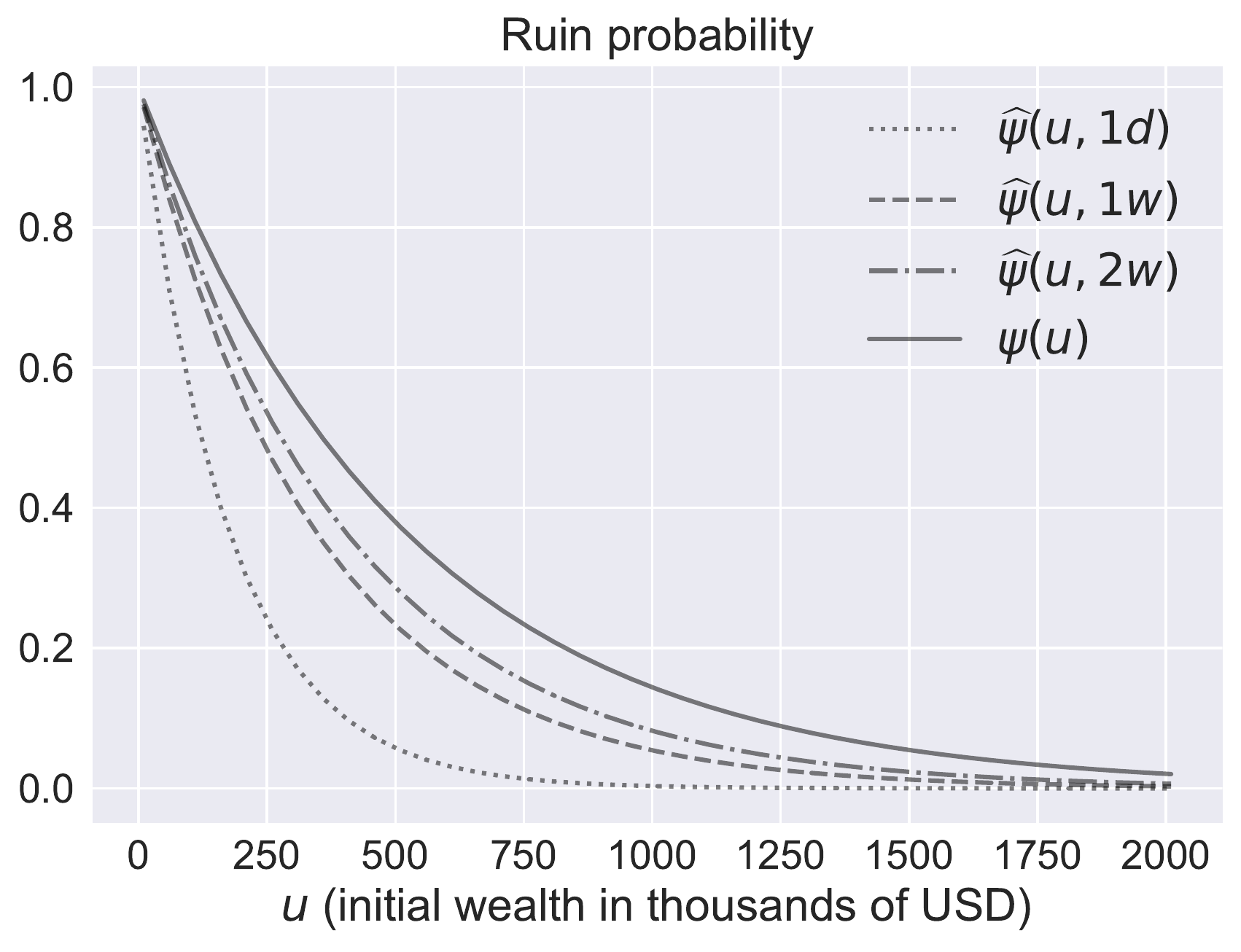}
      \label{sub:pro_rp_u}
                         }
                         \hskip1em
    \subfloat[$\vh(u,t)-u$ as a function of initial wealth $u$ for varying time horizons: (dotted) $t = 1d$, (dashed) $t = 1w$, and (dash-dotted) $t=2w$.]{
      \includegraphics[width=0.45\textwidth, height=5.8cm]{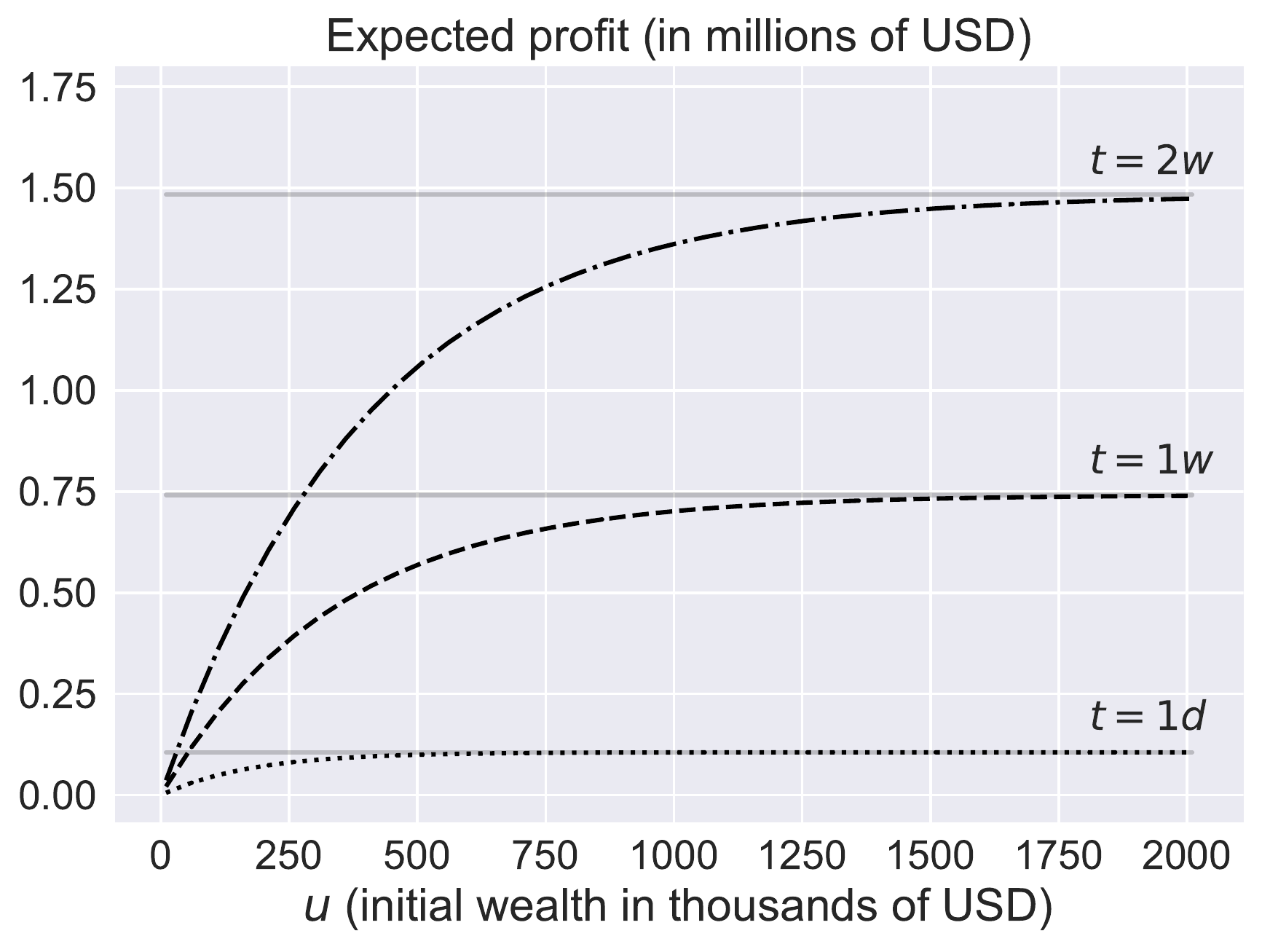}
      \label{sub:pro_rev_u}
                         }
    \caption{Ruin probabilities and expected profit over an exponential time horizon as a function of initial wealth for varying time horizon with hashpower $p=0.1$ and electricity price $\pi_W = 0.06$.}
    \label{fig:pro_ruin_rev_u}
  \end{center}
\end{figure}

\noindent Figure \ref{sub:pro_rev_u} depicts the expected profit over an exponential time horizon with varying mean, using the explicit formula of Proposition \ref{prop:value_function_protocol}. It also contains a comparison with the respective target profitability without the ruin constraint for each case. Note that the latter increases with the time horizon due to the net profit condition. For smaller time horizons, the risk of going to ruin is quickly mitigated when increasing the initial reserve, which makes the convergence of the expected profit towards the targeted one swifter.\\

\noindent Let us now investigate the effect of the electricity price and the hashpower on the expected surplus under ruin constraints. We consider a random time horizon with mean 2 weeks now, and compute the expected surplus for various initial wealth levels. Figure \ref{sub:pro_rev_pkW} plots the respective surplus $\vh(u,t)$ determined in Proposition \ref{prop:value_function_protocol} as a function of the 
electricity price for a fixed hashpower $p=0.1$ (in this plot we do not subtract $u$ in order to show the decline of $\vh(u,t)$ more prominently). Rising electricity prices clearly make mining less profitable, and reduces the surplus to virtually zero for $\pi_W$ around $0.08$ even for quite large initial capital. Figure \ref{sub:pro_rev_p} then displays the expected profit as a function of hashpower, when the electricity price is fixed at $\pi_W = 0.06$. One observes that the expected profit increases almost linearly in the hashpower when the net profit condition holds.    \\

\begin{figure}[h!]
  \begin{center}
    \subfloat[$\vh(u,2w)$ as a function of electricity price $\pi_W$ with hashpower $p=0.1$ for varying initial wealth: (dotted) $u =1.1$, (dashed) $u = 4.1$, and (dash-dotted) $u = 8.1$.]{
      \includegraphics[width=0.45\textwidth]{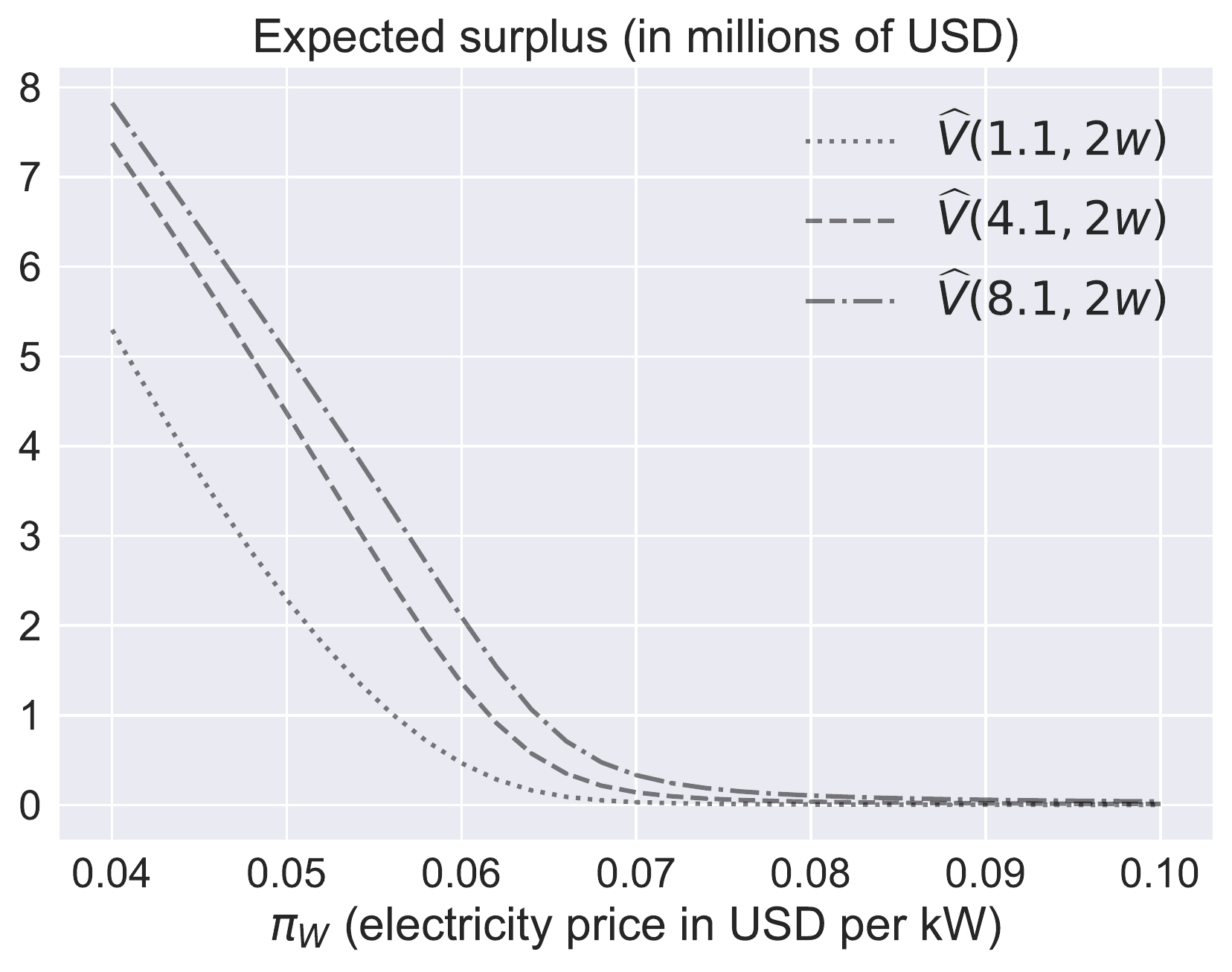}
      \label{sub:pro_rev_pkW}
                         }
                         \hskip1em
    \subfloat[$\vh(u,2w)-u$ as a function of hashpower $p$ with electricity price $\pi_W = 0.06$ for varying initial wealth: (dotted) $u =1.1$, (dashed) $u = 4.1$, and (dash-dotted) $u = 8.1$.]{
      \includegraphics[width=0.45\textwidth, height=5.7cm]{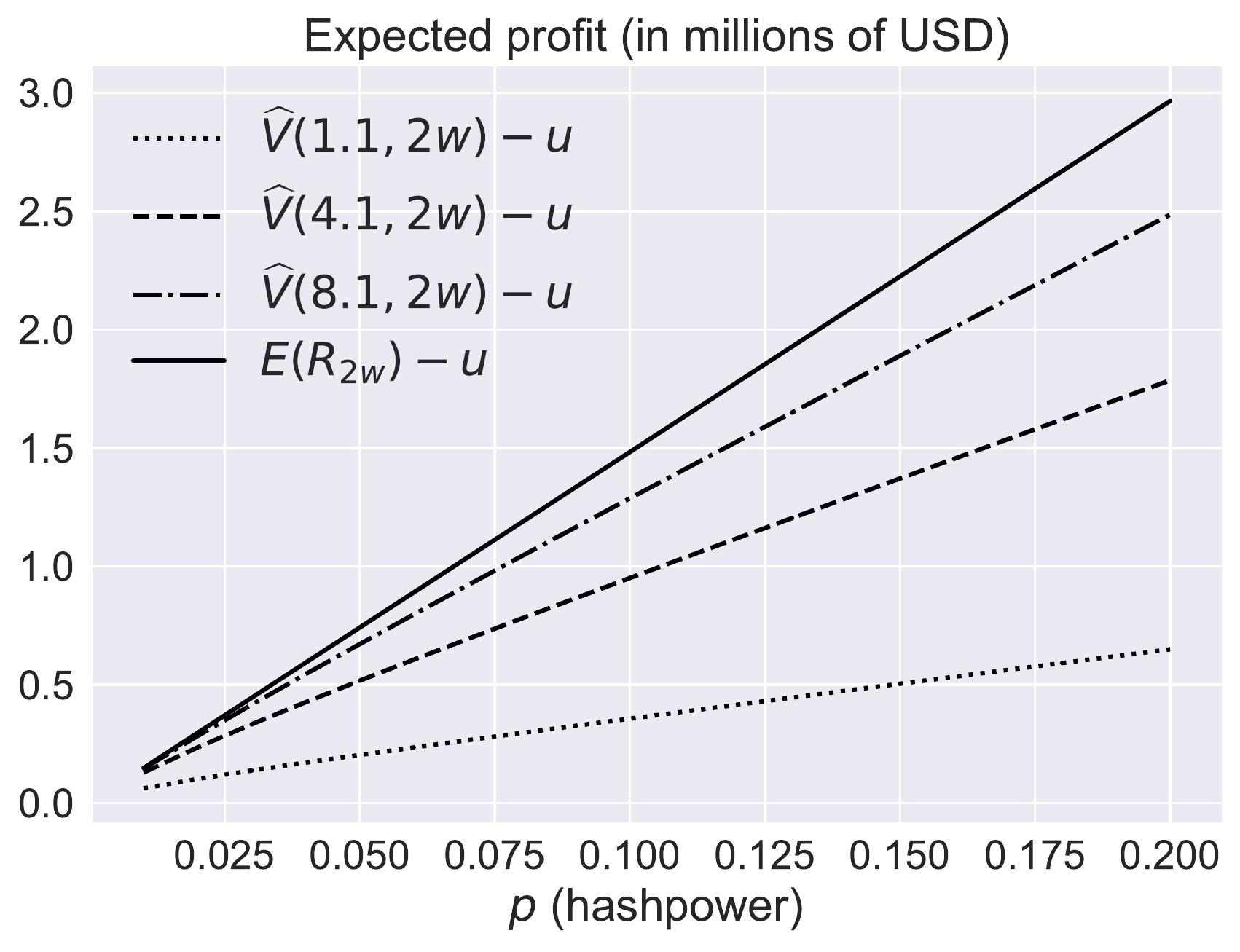}
      \label{sub:pro_rev_p}
                         }
    \caption{Expected surplus as a function of electricity price (left) and expected profit as a function of  hashpower (right) under ruin constraints over an exponential time horizon (mean 2 weeks), for varying initial wealth (in tens of thousands of USD).}
    \label{fig:pro_rev_pkW_p}
  \end{center}
\end{figure}

\noindent The magnitude of the reward for discovering a block makes mining very profitable in expectation, but the rarity of such an event also makes it very risky. Miners who seek a more steady income are then incentivized to form teams, called \textit{mining pools} (Rosenfeld \cite{Ro11}), which reduce the variance of the reward process. Our results may be used to gain insight into how beneficial it is for a miner to join a mining pool. Consider $100$ miners that start looking for blocks in isolation. Each of them owns $1\%$ of the computing power and has an initial capital of amount $\$10,000$. We wish to compare the ruin probability and the expected profit of mining pools of increasing size. We simply assume that when two miners join forces, the resulting pool then has $2\%$ of the overall hashpower and $u^*=\$20,000$ as initial capital. Figure \ref{fig:pro_rev_rp_pool} displays the ruin probability and expected profit per miner over an exponential time horizon as a function of the size of the mining pool for varying time horizon $t\in\{1d, 1w, 2w\}$. The expected profit per miner plotted in Figure \ref{sub:pro_rev_pool} corresponds to the expected profit of the mining pool divided by the number of its members. One clearly can see the advantage of joining a pool, when the risk of ruin is added to the analysis.\\

\begin{figure}[h!]
  \begin{center}
    \subfloat[$\ph(u^*,t)$ as a function of the size of the mining pool for time horizon $t=1d$ (dotted),  $t=1w$ (dashed) and $t = 2w$ (dash-dotted).]{
      \includegraphics[width=0.45\textwidth]{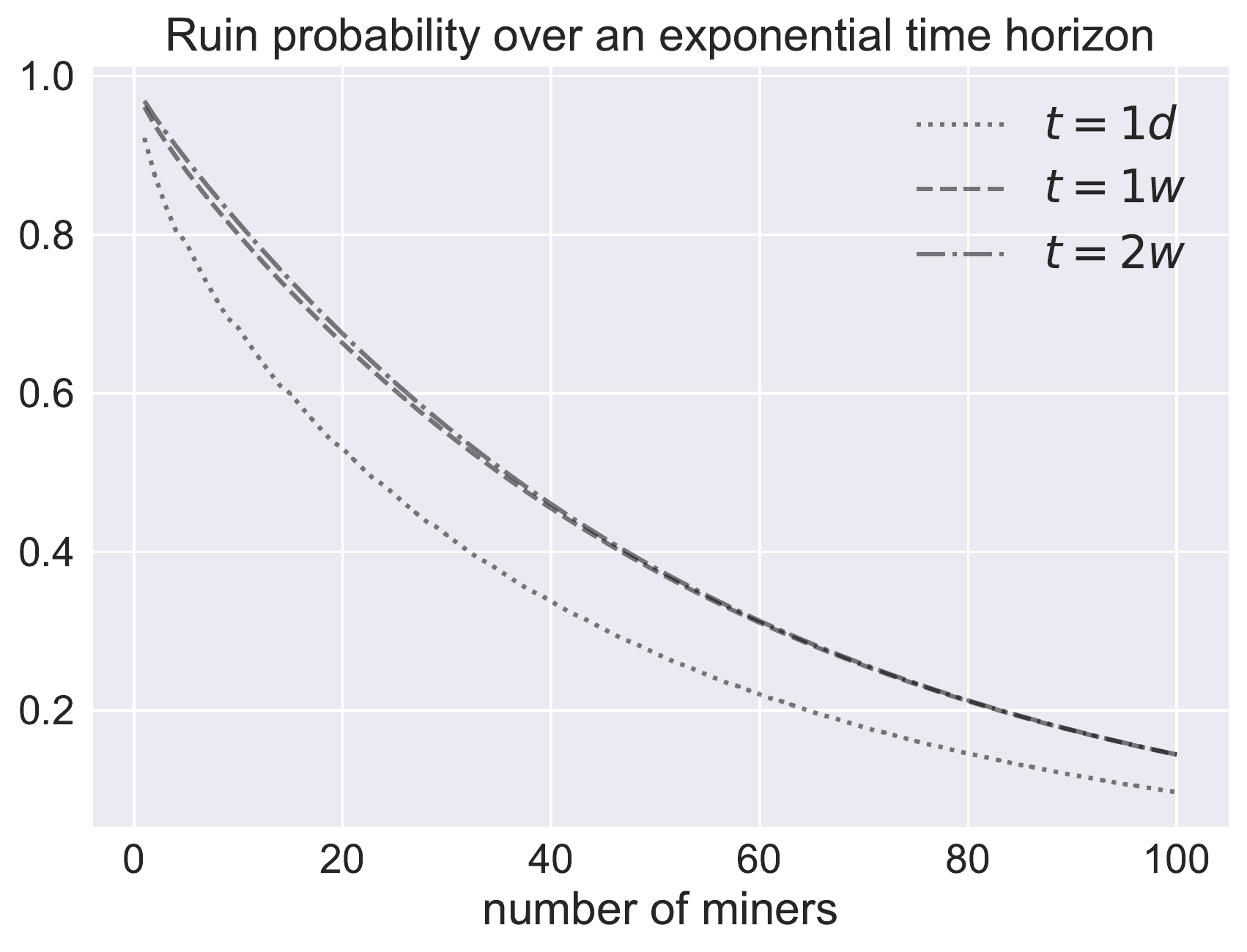}
      \label{sub:pro_rp_pool}
                         }
                         \hskip1em
    \subfloat[$\vh(u^*,t)-u^*$ divided by the number of miners as a function of that number with  $t=1d$ (dotted), $t=1w$ (dashed) and $t = 2w$ (dash-dotted).]{
      \includegraphics[width=0.45\textwidth, height=5.7cm]{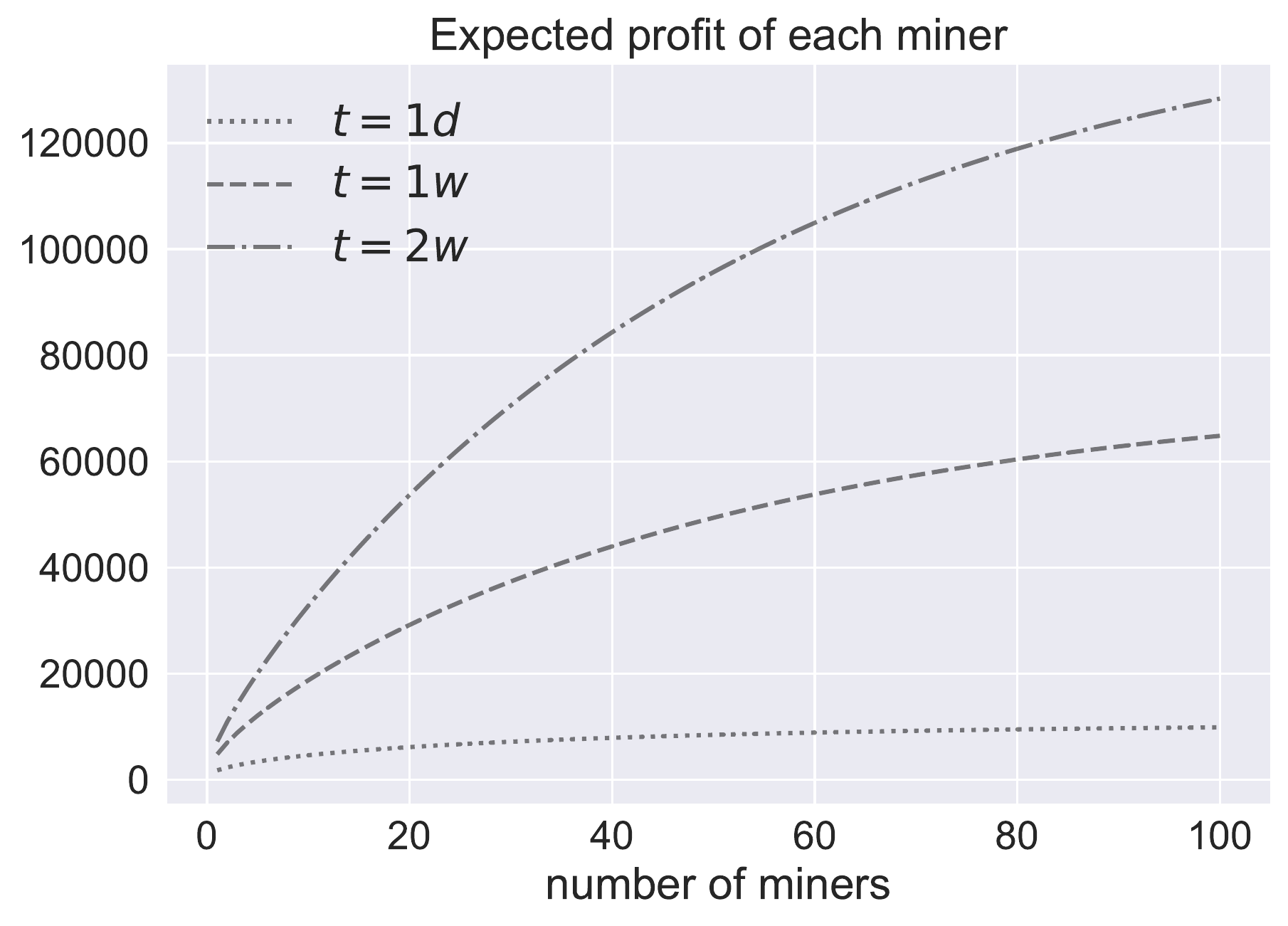}
      \label{sub:pro_rev_pool}
                         }
    \caption{Ruin probability and expected profit per miner over an exponential time horizon as a function of size of the mining pool for varying time horizon. }
    \label{fig:pro_rev_rp_pool}
  \end{center}
\end{figure}

\noindent Note that the proportional reward system considered here is slightly simplistic, as it does not prevent miners from switching the pool, which can be an issue in practice, see for instance Rosenfeld \cite{Ro11} and Lewenberg et al.\ \cite{10.5555/2772879.2773270}. 

\subsection{Expected profit computations for a selfish miner}\label{sub:numerical_illustration_self}
We now turn to the study of the profitability of selfish mining using model \eqref{eq:surplus_selfish}. Assuming again that $c = p\times \pi_W\times W$, the net profit condition \eqref{netpro} is satisfied if the electricity price is bounded by
$$ 
\pi_W \leq \frac{b\lambda}{pW}\frac{qp(1-p)^2 + 4p^2-2p^3}{1+2p-p^2},
$$
where $p$ denotes the hashpower and $q$ is the connectivity parameter. Figure \ref{fig:nbc_selfish_q} displays the net profit condition frontier as a function of hashpower and  electricity price for various values of $q$. For instance, one can read off from the plot that for a connectivity parameter  $q=0.5$ and hashpower $p=0.1$ the net profit condition is satisfied whenever $\pi_W\leq 0.43$.

\begin{figure}[h!]
  \begin{center}

  \includegraphics[width=0.45\textwidth]{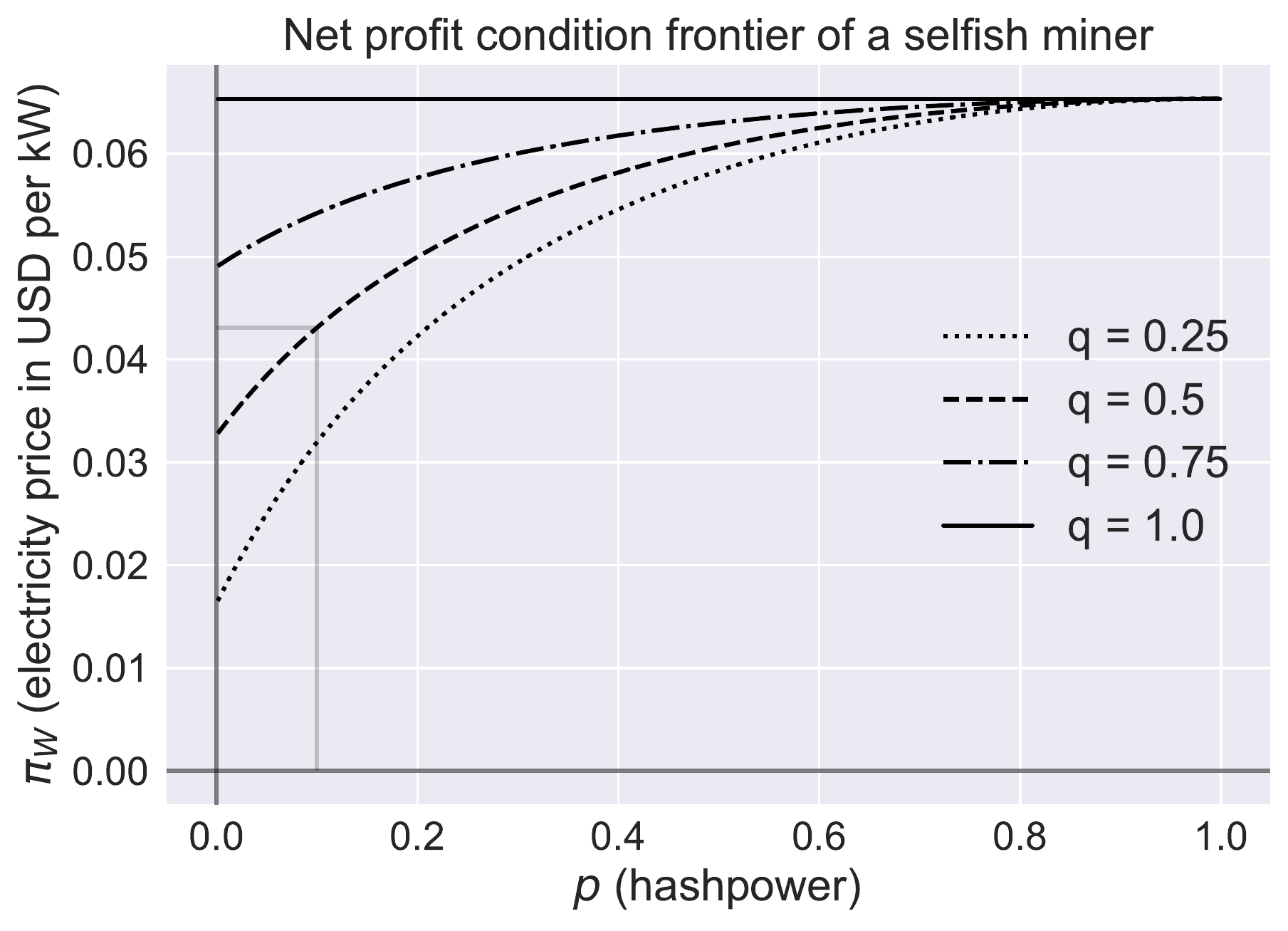}
    \caption{Net profit condition frontier of a selfish miner as a function of electricity price and hashpower for various values of the connectivity parameter $q$.}
    \label{fig:nbc_selfish_q}
  \end{center}
\end{figure}

\noindent Note that the net profit condition for $q=1$ is identical to that of a miner following the protocol, since it refers to the stationary distribution (cf. Remark \ref{rem32} for further details). \\

\noindent Figure \ref{fig:self_rev_u_MC} compares the ruin probability and expected profit of a selfish miner (as determined by Theorem \ref{theo:value_Function_Selfish}) and a miner following the protocol, plotted as a function of initial wealth for $q = 0.5$, $p = 0.1$ and $\pi_W = 0.04$ (so the net profit condition holds, cf.\  Figure \ref{fig:nbc_selfish_q}). 
\begin{figure}[h!]
	\begin{center}
	\subfloat[$\ph(u,t)$ (solid) and $\ph_0(u,t)$ (dashed) as a function of initial wealth $u$.]{
      \includegraphics[width=0.45\textwidth]{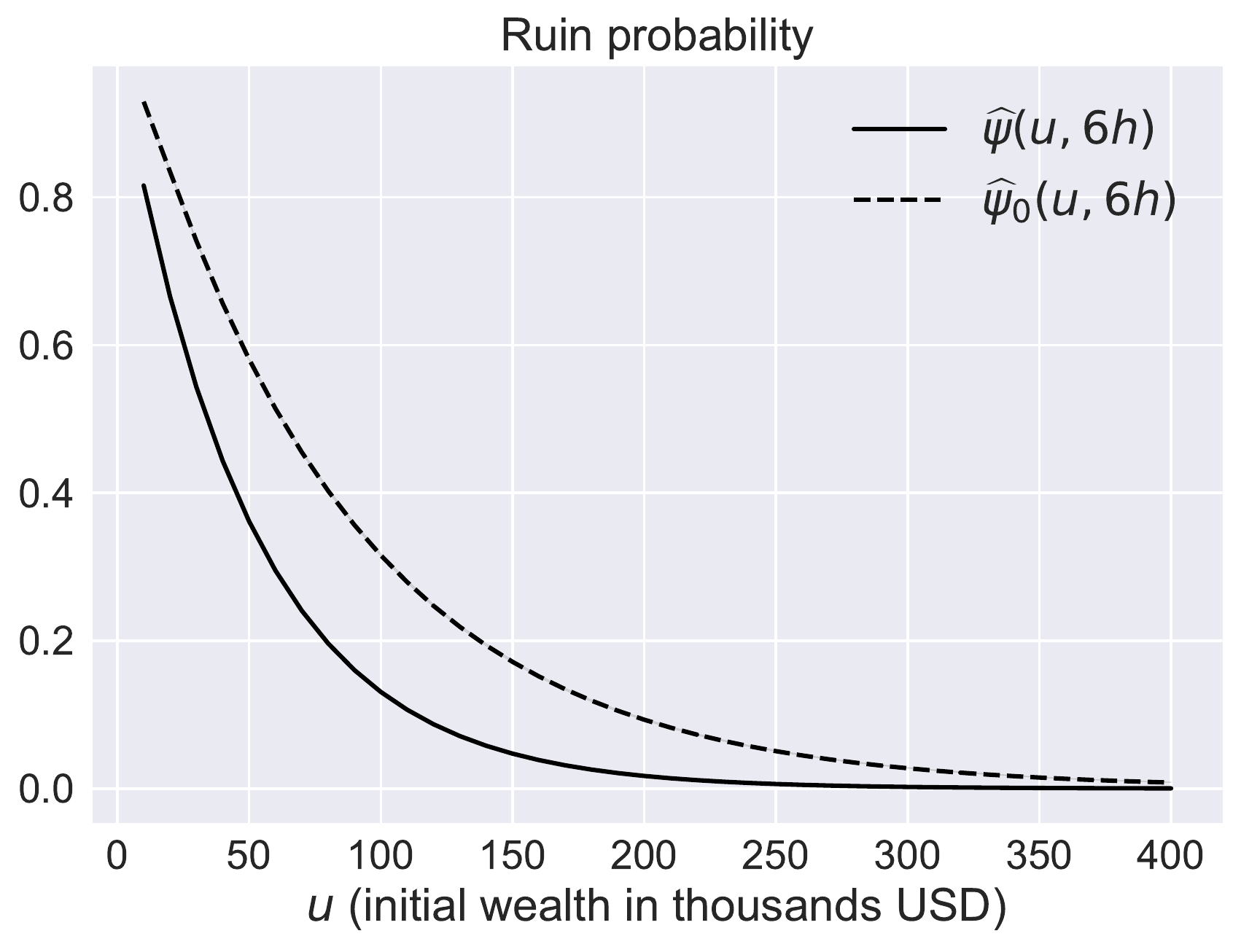}
      \label{sub:self_rp_u_MC}
                         }
                         \hskip1em
    \subfloat[$\vh(u,t)-u$ (solid) and $\vh_0(u,t)-u$ (dashed) as a function of initial wealth $u$.]{
      \includegraphics[width=0.45\textwidth, height=5.7cm]{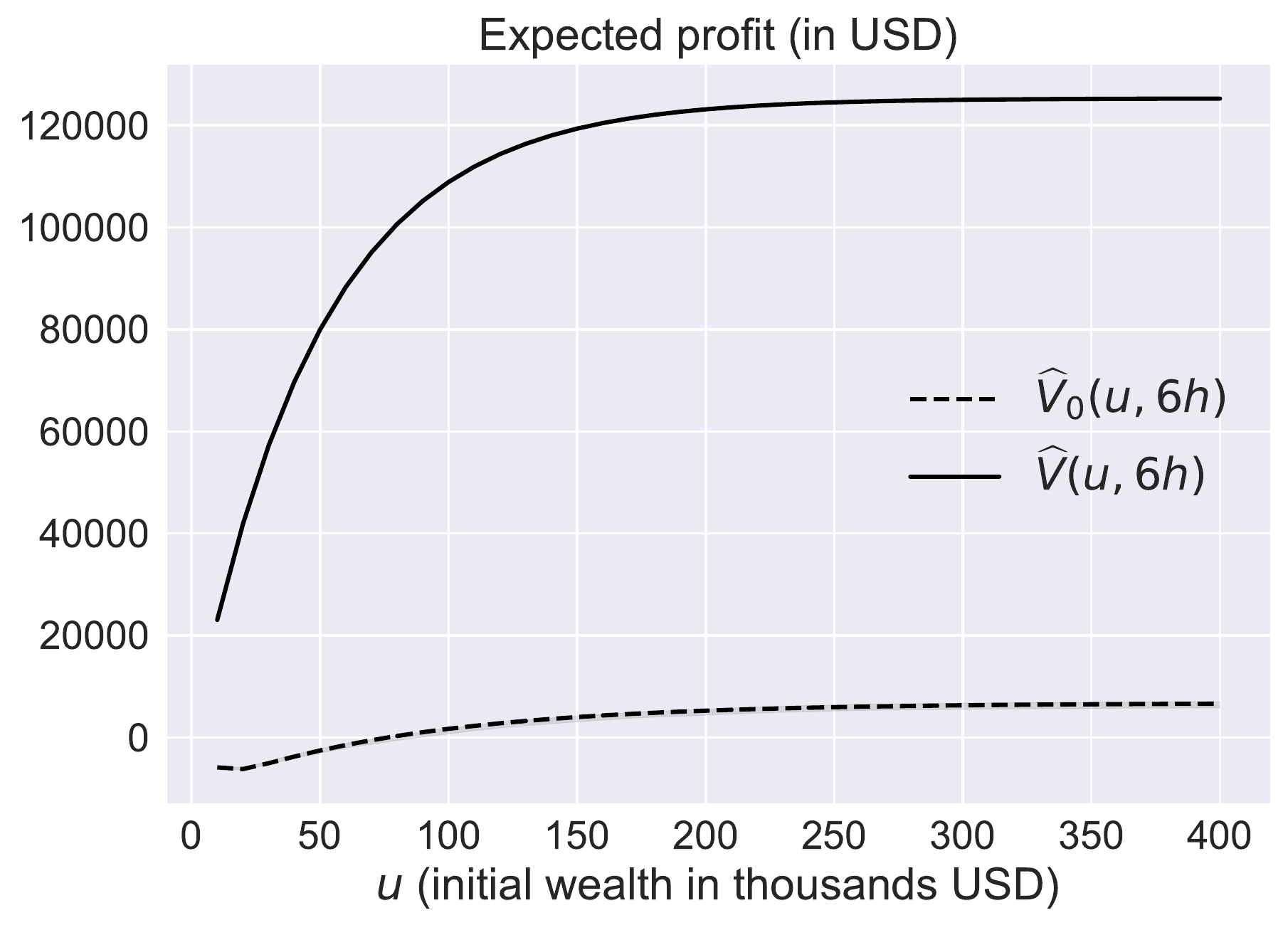}
      \label{sub:self_rev_u_MC}
                         }
	\caption{Ruin probability and expected profit over an exponential time horizon $T\sim\text{Exp}(6h)$ as a function of initial wealth for a miner following the protocol and for a selfish miner with electricity price $\pi_W = 0.04$, hashpower $p=0.1$, and connectivity $q = 0.5$, together with the confidence bands of a Monte Carlo simulation.}
	\label{fig:self_rev_u_MC}
	\end{center}
\end{figure}
Here the exponential time horizon is chosen to have a mean of 6 hours. The plot also gives (the very narrow) 95\% confidence bands of an independent Monte Carlo simulation with 250'000 sample paths for $\ph_0(u)$ and $\vh_0(u,t)$, which nicely shows that the exact formulas are within these bounds. Note the substantial gap between the honest and the selfish miner's ruin probability in Figure \ref{sub:self_rp_u_MC} and expected profit in Figure \ref{sub:self_rev_u_MC}.\\

\noindent Let us now look into the impact of the connectivity on the expected profit of a selfish miner. The connectivity depends on the topography of the network as well as on the information propagation delay between the nodes. Göbel et al.\ \cite{Goebel2016} showed that the connectivity strongly influences the relative revenue of selfish miners. Figure \ref{fig:self_rev_u_pkW_p} investigates the impact of $q$ on the absolute revenue of a selfish miner. In Figure \ref{sub:self_rev_u}, we plot $\vh_0(u,t)-u$ as a function of initial wealth $u$ for various levels of $q$ (the choice of the other parameters ensures that the net profit condition holds even under the lowest connectivity regime $q=0.25$). The plot also confirms the relative proximity of the value functions when following the protocol and when  withholding blocks with connectivity $q=1$, respectively (cf.\ Remark \ref{rem32}). However, for small values of $u$ the difference can still be substantial in absolute terms, see Figure \ref{sub:diff_rev_pro_rev_u} where we plot the difference $\vh(u,t)-\vh_0(u,t)$ directly.\\ 

\begin{figure}[h!]
  \begin{center}
    \subfloat[$\vh_0(u,2w)-u$ as a function of initial wealth $u$ with haspower $p=0.1$ and electricity price $\pi_W = 0.03$ for varying connectivity level: $q=0.25$ (dotted), $q = 0.5$ (dashed), $q=0.75$ (dash-dotted), $q=1$ (solid).]{
      \includegraphics[width=0.45\textwidth]{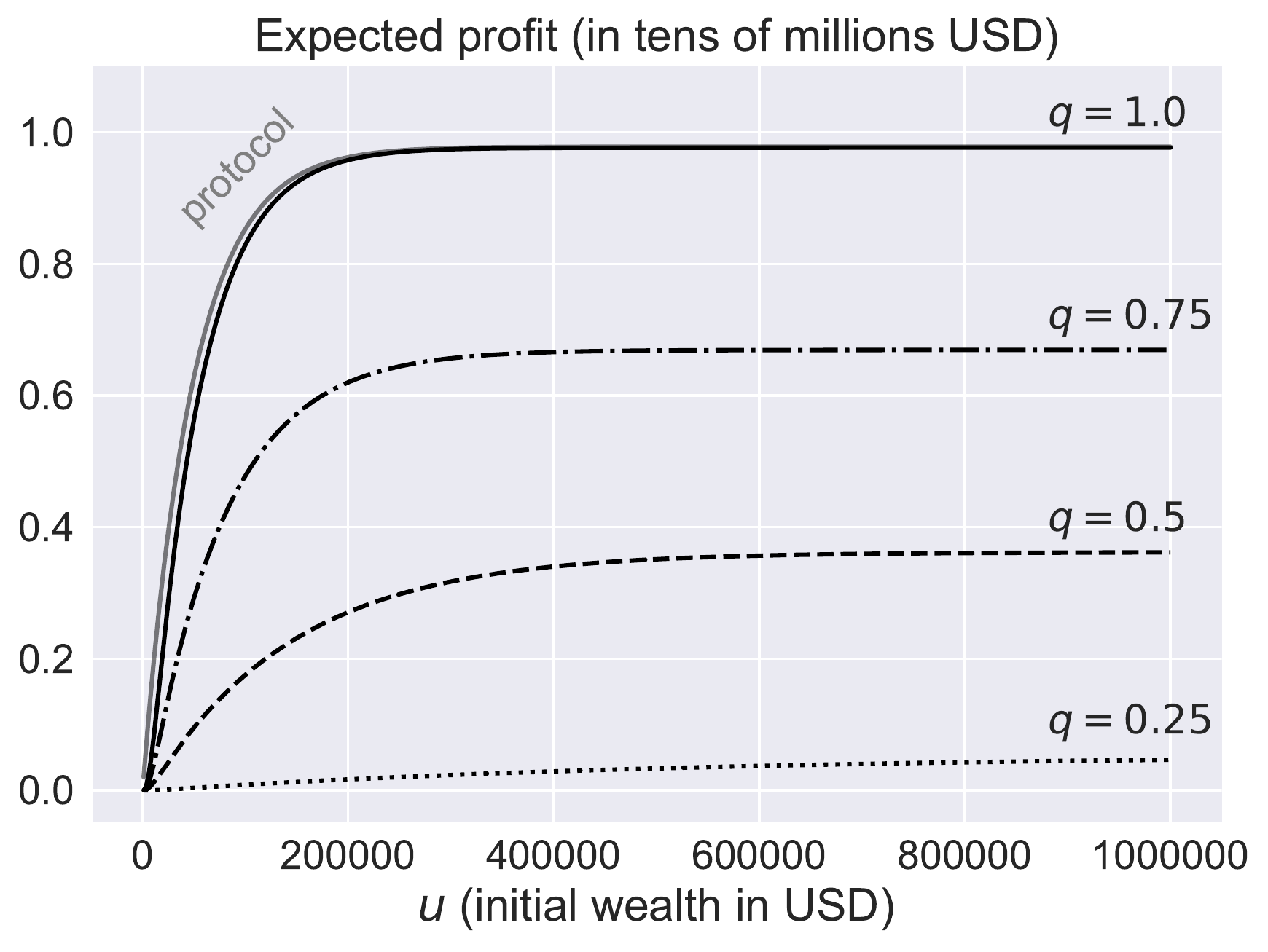}
      \label{sub:self_rev_u}
                         }
                         \hskip1em
    \subfloat[$\vh(u,2w)-\vh_0(u,2w)$ as a function of initial wealth $u$ with hashpower $p=0.1$, electricity price $\pi_W = 0.03$ and connectivity $q=1$.]{
      \includegraphics[width=0.45\textwidth]{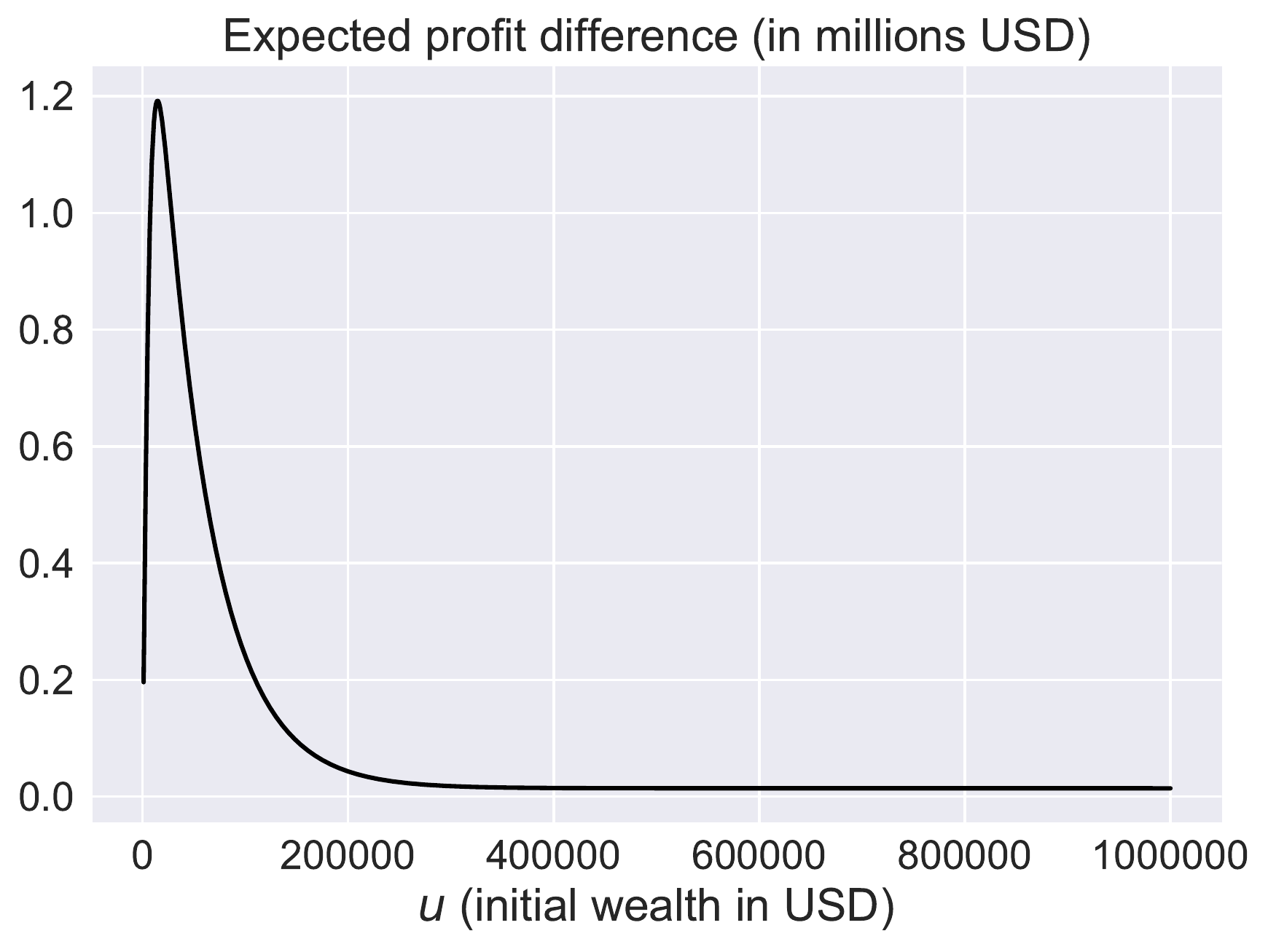}
      \label{sub:diff_rev_pro_rev_u}
                         }
    
    \caption{Expected profit of a selfish miner over an exponential time horizon (mean $2$ weeks) as a function of initial wealth (left) compared to the expected profit of a miner following the protocol (right).}
    \label{fig:self_rev_u_pkW_p}
  \end{center}
\end{figure} 

\noindent Figure \ref{sub:self_rev_pkW} plots the expected surplus under ruin constraints as a function of the electricity price, and in Figure \ref{sub:self_rev_p} we see how the choice of the hashpower influences the expected profit $\vh_0(u,t)-u$ of a selfish miner.\\

\begin{figure}[h!]
  \begin{center}
    \subfloat[$\vh_0(u,2w)$ as a function of electricity price $\pi_W$ with hashpower $p=0.1$ and initial wealth $u = \$410,000$.]{
      \includegraphics[width=0.45\textwidth,height=5.3cm]{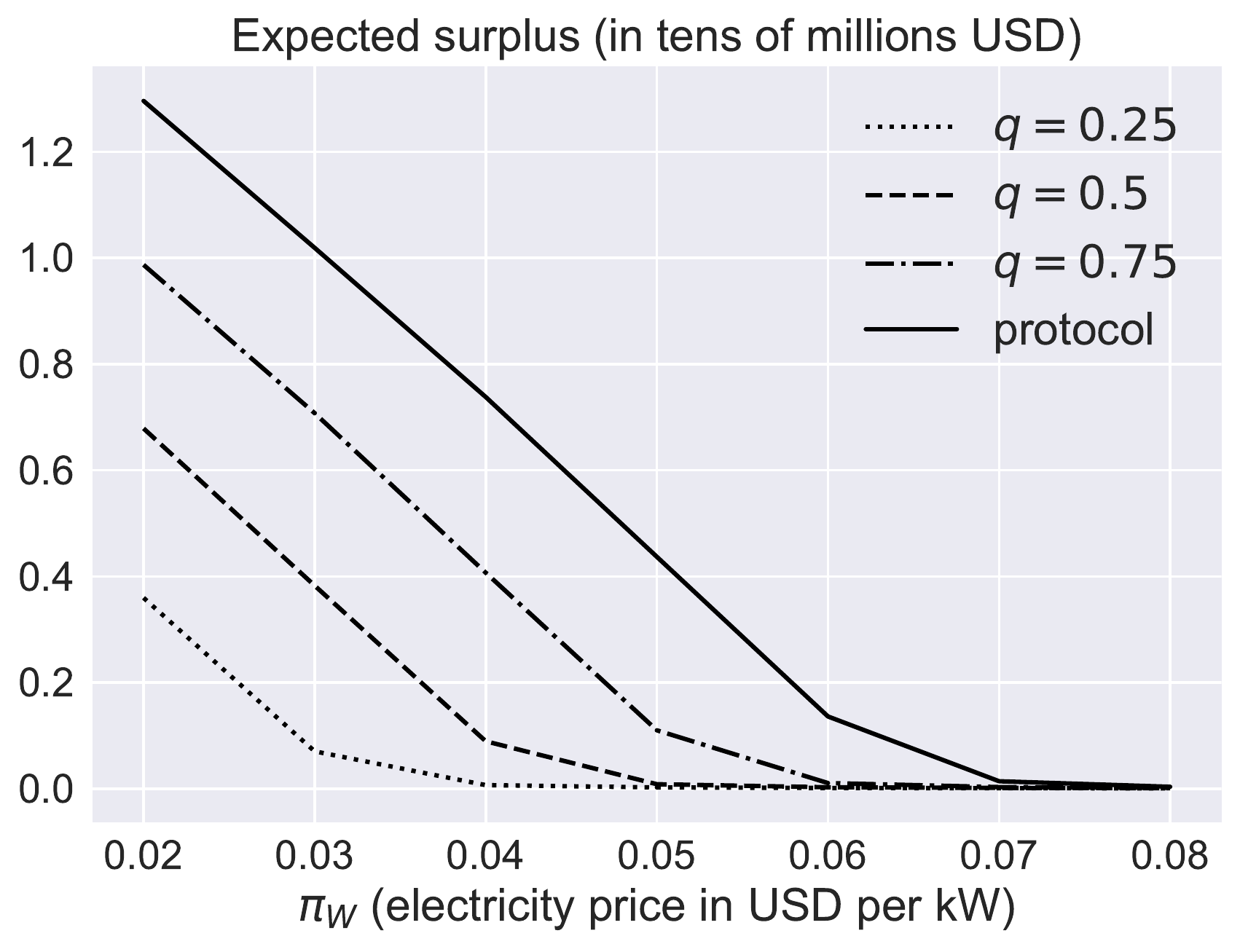}
      \label{sub:self_rev_pkW}
                         }
                         \hskip1em
    \subfloat[$\vh_0(u,2w)-u$ as a function of hashpower $p$ with electricity price $\pi_W = 0.03$ and initial wealth $u = \$410,000$.]{
      \includegraphics[width=0.45\textwidth]{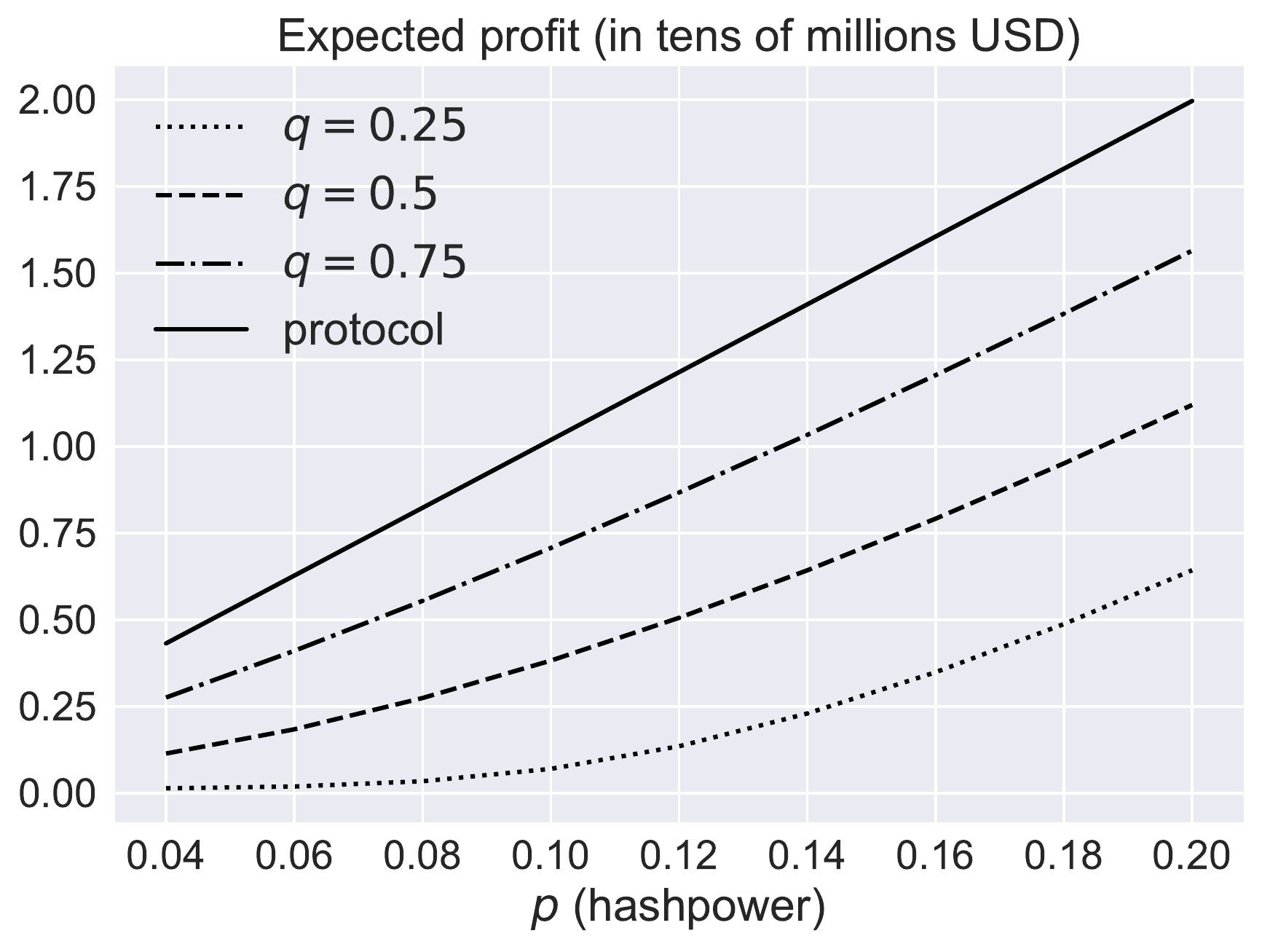}
      \label{sub:self_rev_p}
                         }
    \caption{Expected profit of a selfish miner over an exponential time horizon (mean $2$ weeks) as a function of electricity price (left) and hashpower (right) for varying connectivity level $q$, and a comparison with a miner following the protocol (solid).}
    \label{fig:self_rev_u_pkW_p}
  \end{center}
\end{figure} 

\noindent We observe that the influence of the initial wealth, electricity price and hashpower on the expected profit of a selfish miner is similar (in terms of shape) to that of a miner following the protocol, but -- compared in absolute terms -- withholding blocks leads to substantial losses in expected profit even in high connectivity regimes like $q=0.75$. This confirms that a miner is better off by following the protocol in terms of 'absolute' (as opposed to 'relative' in Eyal and Sirer \cite{EySi18}) revenue under solvency constraints. In the next section we want to check wether this also holds true when taking into account an adjustment in the cryptopuzzle difficulty.

\section{Expected profit when including a difficulty adjustment}\label{sec:diffi}
Mining (at least on the bitcoin blockchain) boils down to drawing random numbers, referred to as hashes uniformly inside the set $\{0,\ldots, 2^{256}-1\}$. The block is mined if the computed hash is smaller than some target $L$. The calibration of $L$ allows to maintain a steady flow of blocks in the blockchain. As each trial (of the system) for mining a block is independent of the others and leads to a success with very small probability, the overall number of successes is binomially distributed and will be very well approximated by a Poisson random variable. Hence the Poisson process assumption for the occurrence of mined blocks is in fact very natural. In the bitcoin blockchain, the target $L$ is set so as to ensure that $6$ blocks are generated per hour. The target may be estimated by comparing $2^{256}$ to the number of hashes computed by the network in ten minutes. On January 1, 2020, the network was computing $97.01$\footnote{Source: \href{https://www.blockchain.com/}{blockchain.com}} exahashes per second. The difficulty is then estimated by $L = 2^{256}/(H/6),$ where $H= 97.01 \times 10^{18} \times 3600$ is the number of hashes computed per hour by the network. The difficulty is adjusted every $2,016$ blocks by changing the target $L$ to 
$$
L^\ast = L\times \frac{t^\ast}{336},
$$
where $t^\ast$ is the time (in hours) it took to mine $2,016$ blocks (here $2016/6=336$ hours is the time it should have taken to mine $2,016$ blocks). The difficulty adjustment was studied in Bowden et al.\ \cite{bowden2018block} and led them to conclude that the block arrival process is well captured by a non-homogeneous Poisson process. Following their terminology, we refer to the time elapsed between two difficulty adjustments as a \textit{segment}.\\

\noindent We now want to quantitatively address whether selfish mining is worthwhile when considering the possibly implied adjustment of the cryptopuzzle difficulty. We do so in a simplified setup, where only Sam may switch between selfish mining and following the protocol, whereas everyone else follows the protocol. Concretely, we compute the expected profit of Sam over two segments when
\begin{enumerate}
\item[(i)]  he is following the protocol during both segments,
\item[(ii)] he applies selfish mining during the first segment and resumes to follow the protocol during the second segment. 
\end{enumerate}

\noindent For (i), we compute the expected profit using the result of Proposition \ref{prop:value_function_protocol} by setting the average time horizon to $t = 672$ (the number of hours in four weeks). We assume that the arrival intensity of the blocks in the blockchain remains unchanged over the two segments with $\lambda = 6$, as does the cryptopuzzle difficulty $L$.\\

\noindent For (ii), we proceed as follows. Selfish mining slows down the pace at which the blocks are added to the blockchain, and the number of blocks wasted exactly corresponds to the number of passages of the Markov chain $(Z_n)_{n\geq0}$ through the state $0^\ast$. We know that once stationarity is reached, the probability for  $(Z_n)_{n\geq0}$ to be in state $0^\ast$ is 
$$
\frac{p(1-p)}{1+2p-p^2}.
$$
We therefore approximate the arrival process in this first segment by a homogeneous Poisson process with intensity 
 $$
\lambda_1 = \lambda \times\left(1- \frac{p(1-p)}{ 1+2p-p^2}\right).
$$ 
When blocks are being withheld, the average time required to mine $2,016$ blocks increases from $336$ to $t_1 = 2016/\lambda_1$. We hence compute the expected surplus over the first segment using Theorem \ref{theo:value_Function_Selfish} with time horizon $t_1$. Selfish mining during the first segment then leads to a downward adjustment of the cryptopuzzle difficulty  to $L_2 = L\times{t_1}/{336}$ that will be in force during the second segment. As the miner resumes following the protocol on that second segment, the block arrival process becomes again a proper homogeneous Poisson process with intensity $\lambda_2$ to be determined as follows. Let $H$ be the number of hashes computed per hour (hashrate) by the network. The number of blocks mined per hour is then given by 
$$
\lambda_2 = \frac{2^{256}}{L_2\times H}.
$$
The miner's surplus over the second segment is now computed using the formula of Proposition \ref{prop:value_function_protocol} using as initial wealth the miner's expected surplus over the first segment, a block arrival intensity equal to $\lambda_2$ and a reduced time horizon equal to $t_2=2016/\lambda_2$.\\  

\noindent We consider a miner who owns a share $p = 0.1$ of the network computing power. If he follows the protocol, then the net profit condition holds if $\pi_W<0.065$. If he withholds blocks on the first segment, then the net profit condition frontier depends on the electricity price and the hashpower provided that his connectivity is fixed, for instance set to $q=0.5$. Figure \ref{fig:nbc_segments} shows the net profit condition frontiers for a selfish miner on Segment $1$ who resumes to follow the protocol on Segment $2$.\\

\begin{figure}[h!]
  \begin{center}
    \subfloat[Net profit condition frontier on Segment 1]{
      \includegraphics[width=0.45\textwidth]{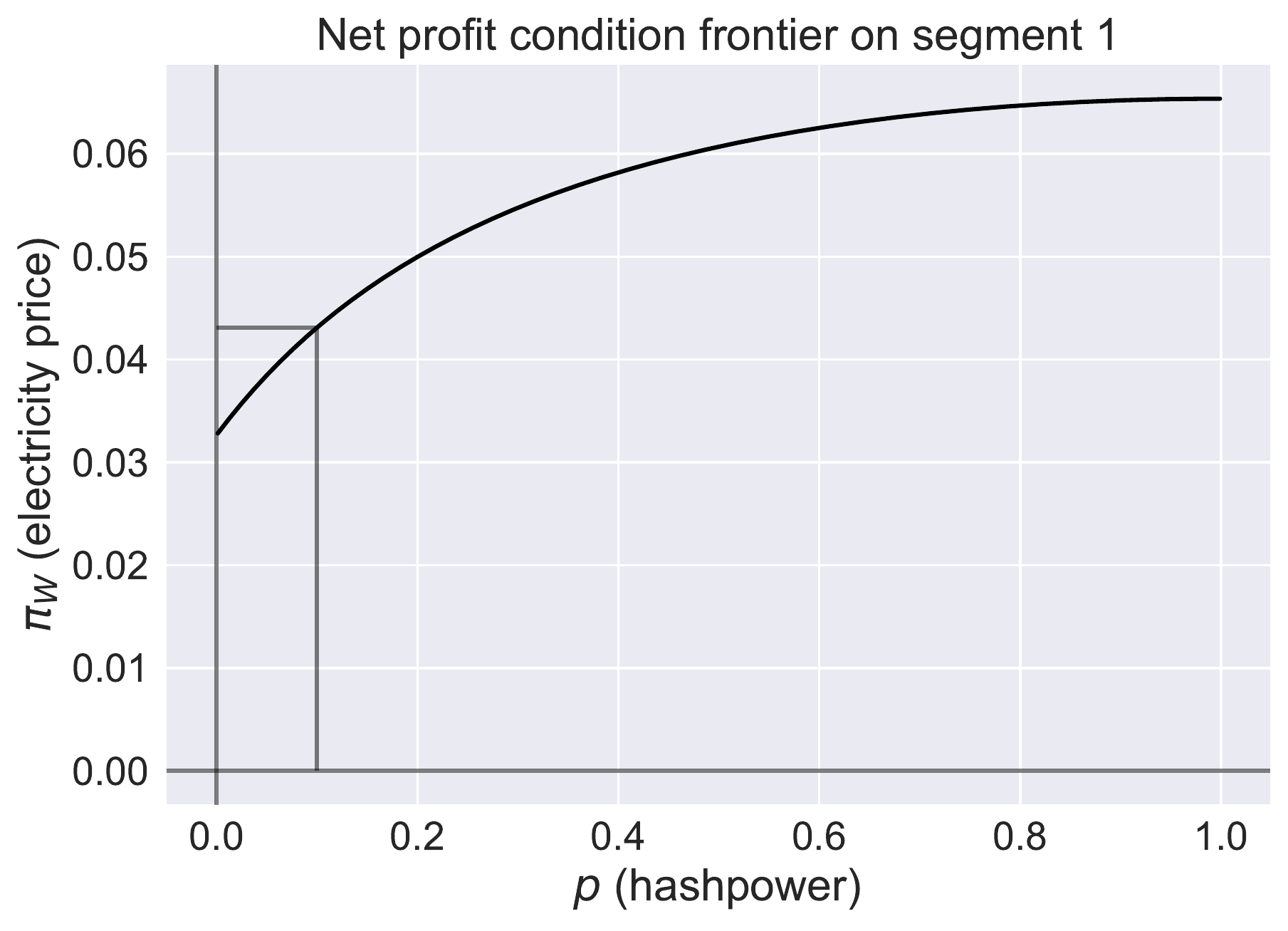}
      \label{sub:nbc_segment1}
                         }
                         \hskip1em
    \subfloat[Net profit condition frontier on Segment 2.]{
      \includegraphics[width=0.45\textwidth]{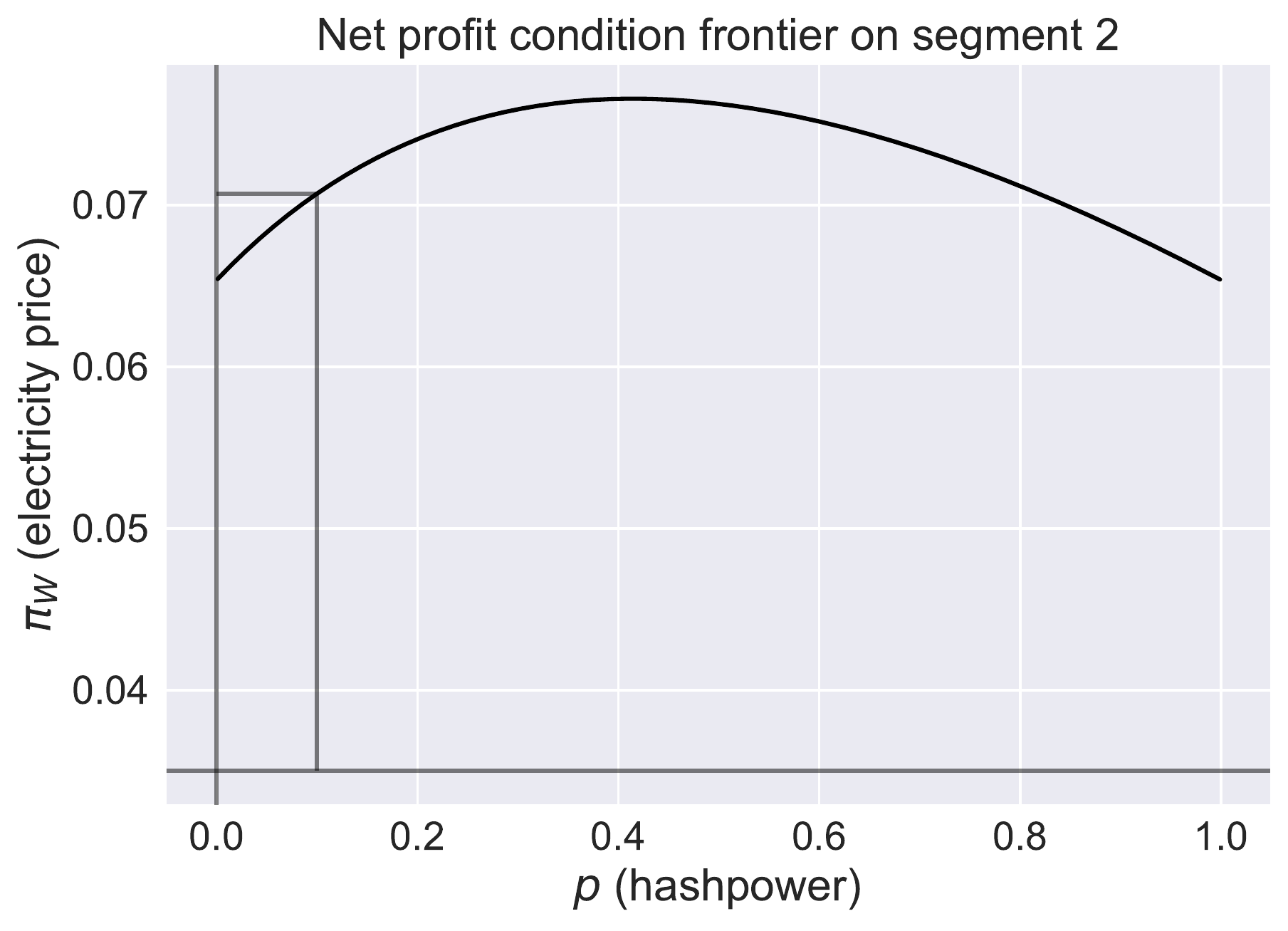}
      \label{sub:nbc_segment2}
                         }
    \caption{Net profit condition frontier on Segment 1 and 2 for a selfish miner.}
    \label{fig:nbc_segments}
  \end{center}
\end{figure}

\noindent In absence of ruin considerations, selfish mining on Segment $1$ is profitable only if the price of electricity is lower than $0.043$, see Figure \ref{sub:nbc_segment1}. Only when the selfish miner owns the totality of the hashpower, selfish mining is as profitable as following the protocol. On Segment $2$, the profitability is always greater than when following the protocol. The profitability on Segment $2$ holds in our case if the electricity price is lower than $0.071$, see Figure \ref{sub:nbc_segment2}. It is interesting to note that by increasing the hashpower beyond a certain threshold we actually lower the profitability during Segment $2$. At higher hashpower levels, the probability of the Markov chain $(Z_n)_{n\geq 0}$ visiting state $0^\ast$ becomes small. In that case fewer blocks are being wasted, which in turn reduces the downward adjustment of the cryptopuzzle difficulty and hence the profitability during Segment 2.\\

\noindent Figure \ref{fig:rev_difficulty_adjusment} displays the expected profit of a selfish miner and a miner following the protocol as a function of initial wealth for a range of electricity prices.
\begin{figure}[h!]
  \begin{center}
      \subfloat[$\pi_W = 0.04$.]{
      \includegraphics[width=0.45\textwidth]{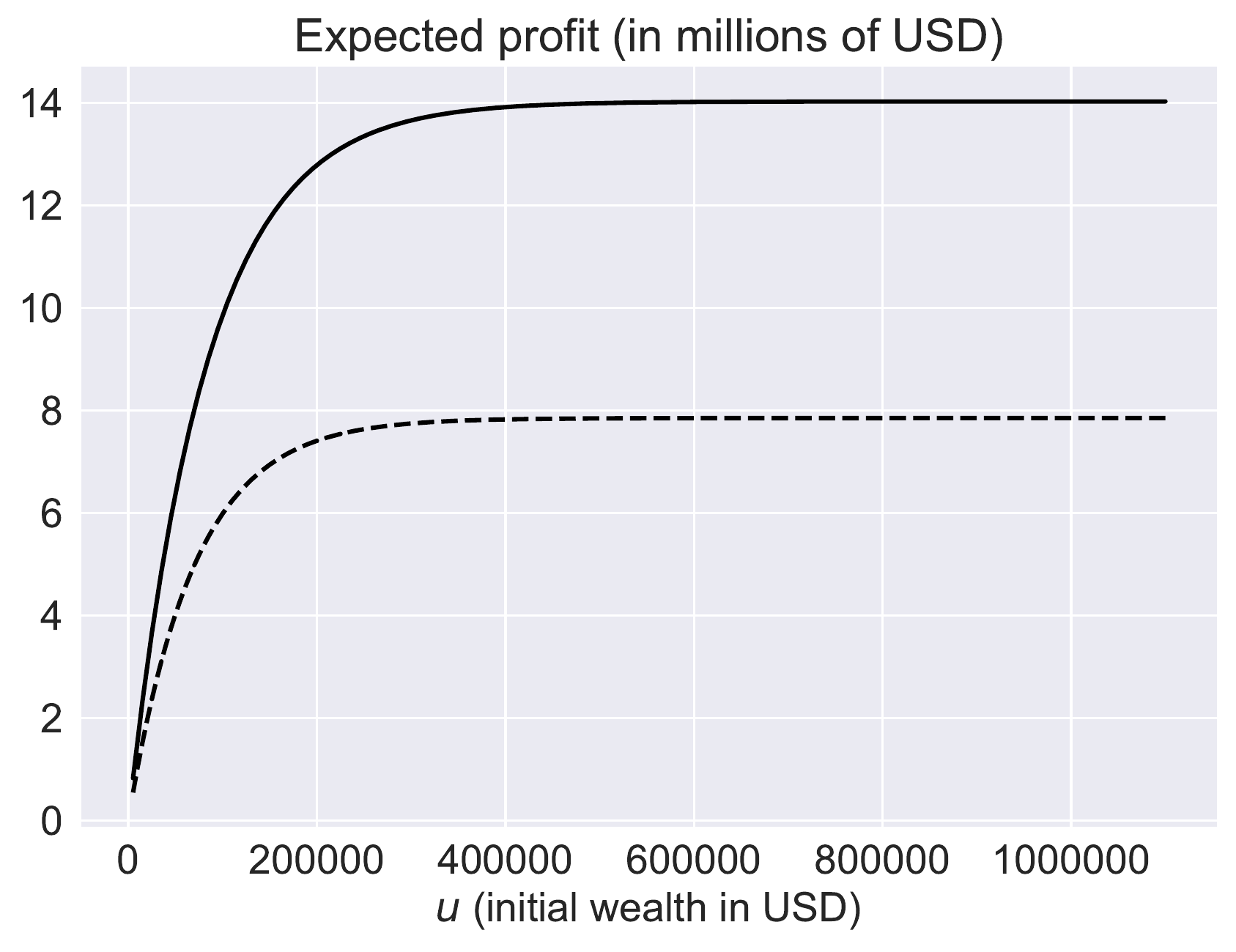}
      \label{sub:rev_difficulty_adjusment_4}
                         }
                         \hskip1em
    \subfloat[$\pi_W = 0.05$.]{
      \includegraphics[width=0.45\textwidth]{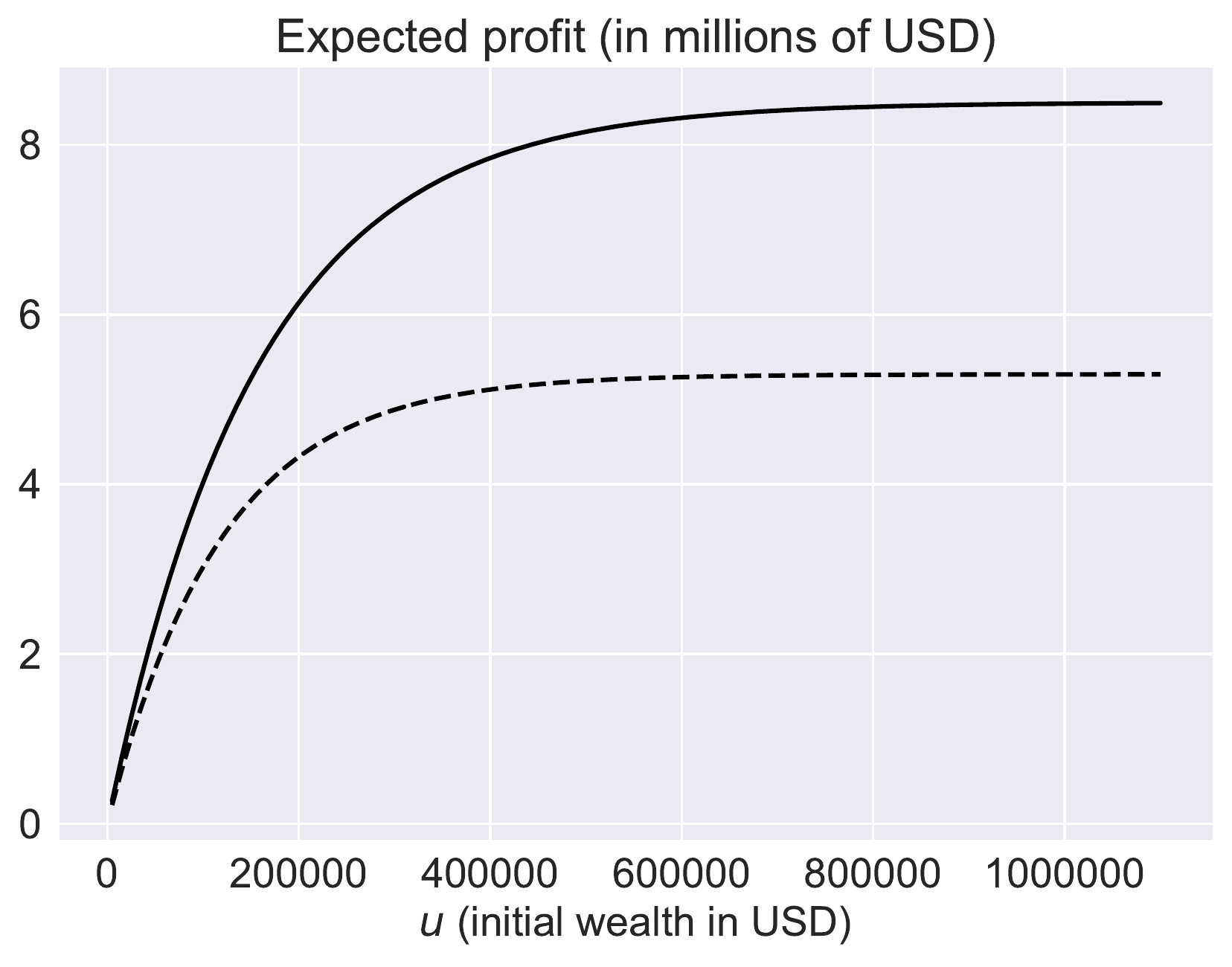}
      \label{sub:rev_difficulty_adjusment_5}
                         }
                         \hskip1em
       \subfloat[$\pi_W = 0.06$.]{
      \includegraphics[width=0.45\textwidth]{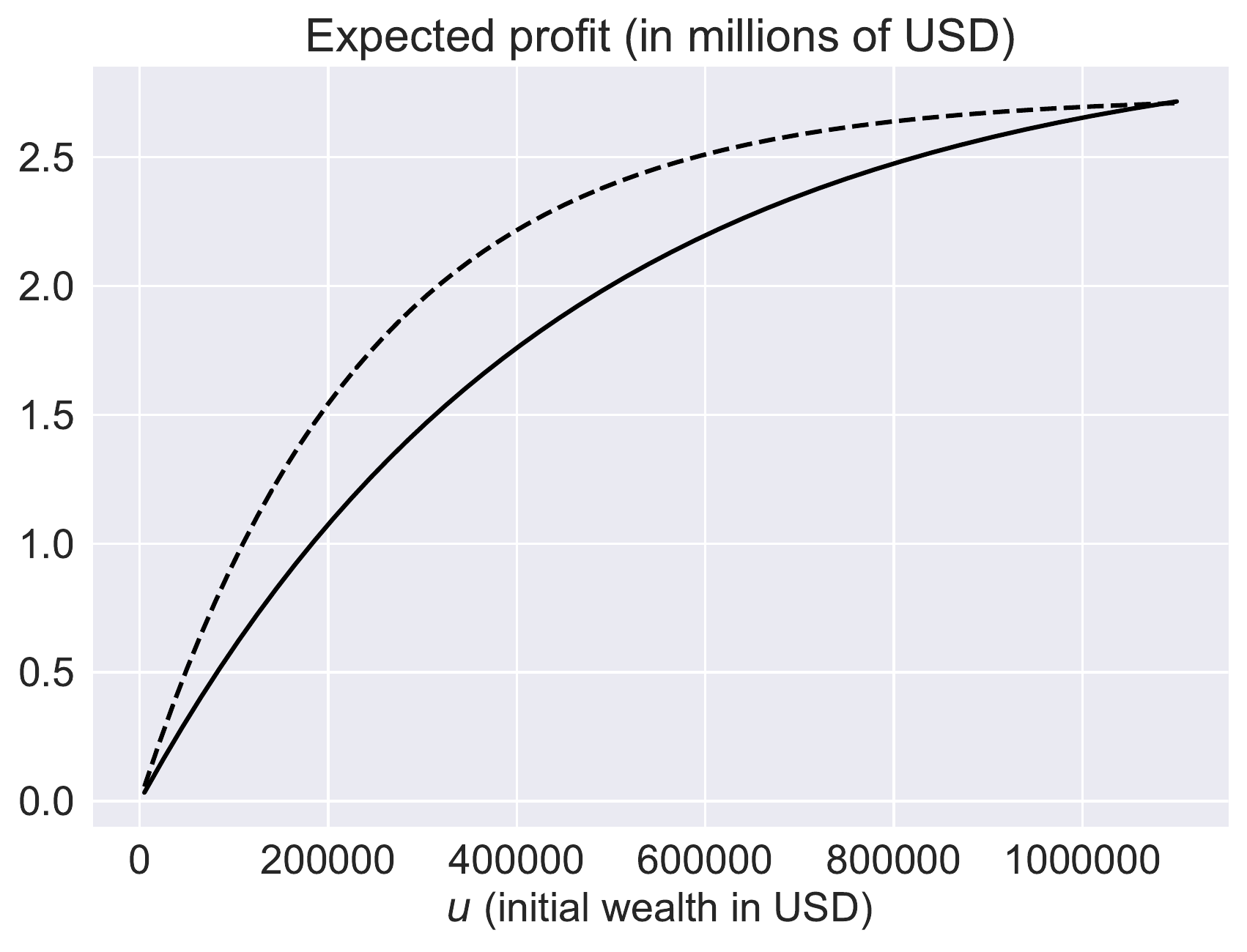}
      \label{sub:rev_difficulty_adjusment_6}
                         }
                         \hskip1em
    \subfloat[$\pi_W = 0.07$.]{
      \includegraphics[width=0.45\textwidth]{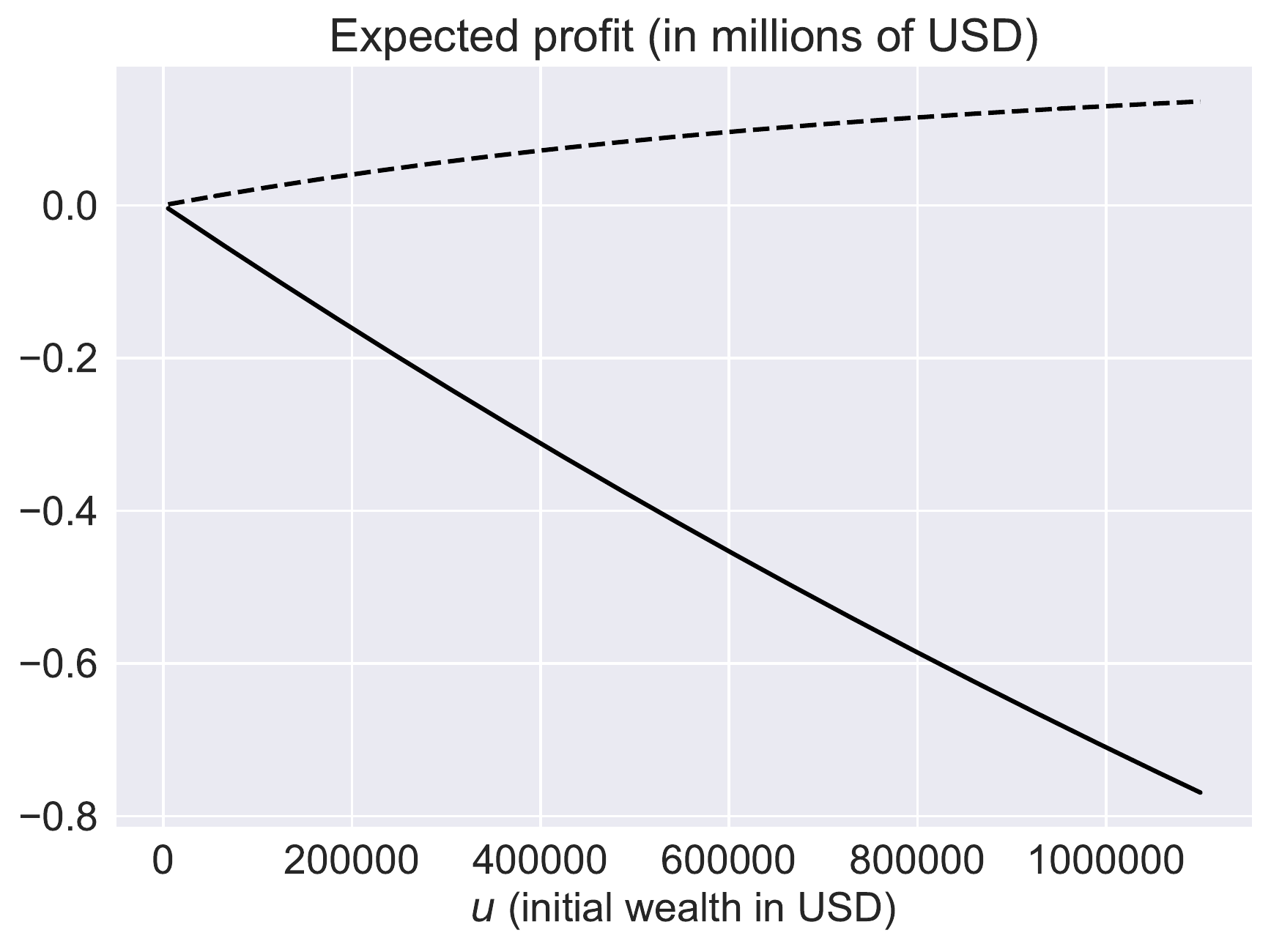}
      \label{sub:rev_difficulty_adjusment_7}
                         }
                         \hskip1em
       \subfloat[$\pi_W = 0.08$.]{
      \includegraphics[width=0.45\textwidth]{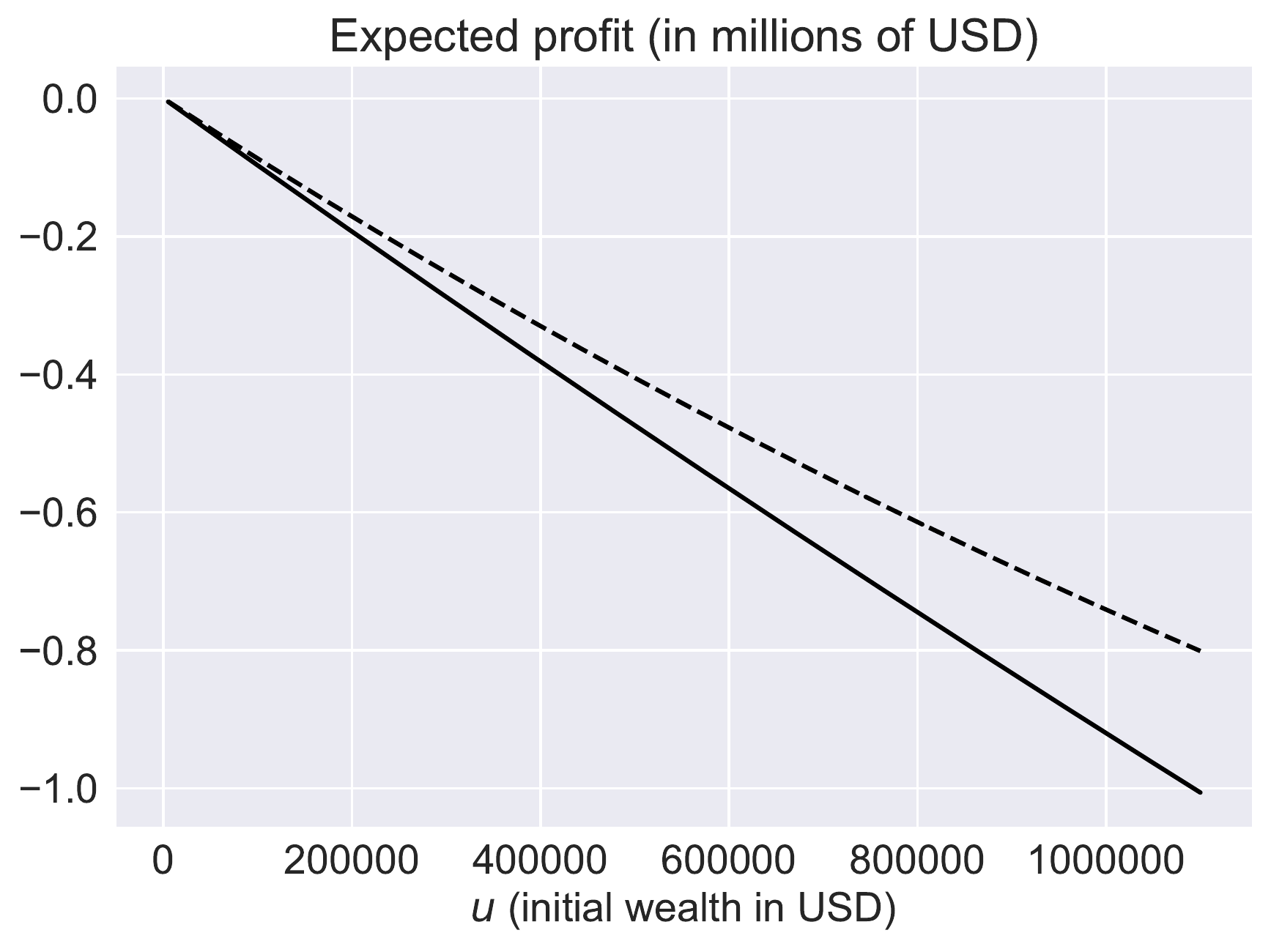}
      \label{sub:rev_difficulty_adjusment_8}
                         }
                         \hskip1em
    \subfloat[$\pi_W = 0.09$.]{
      \includegraphics[width=0.45\textwidth]{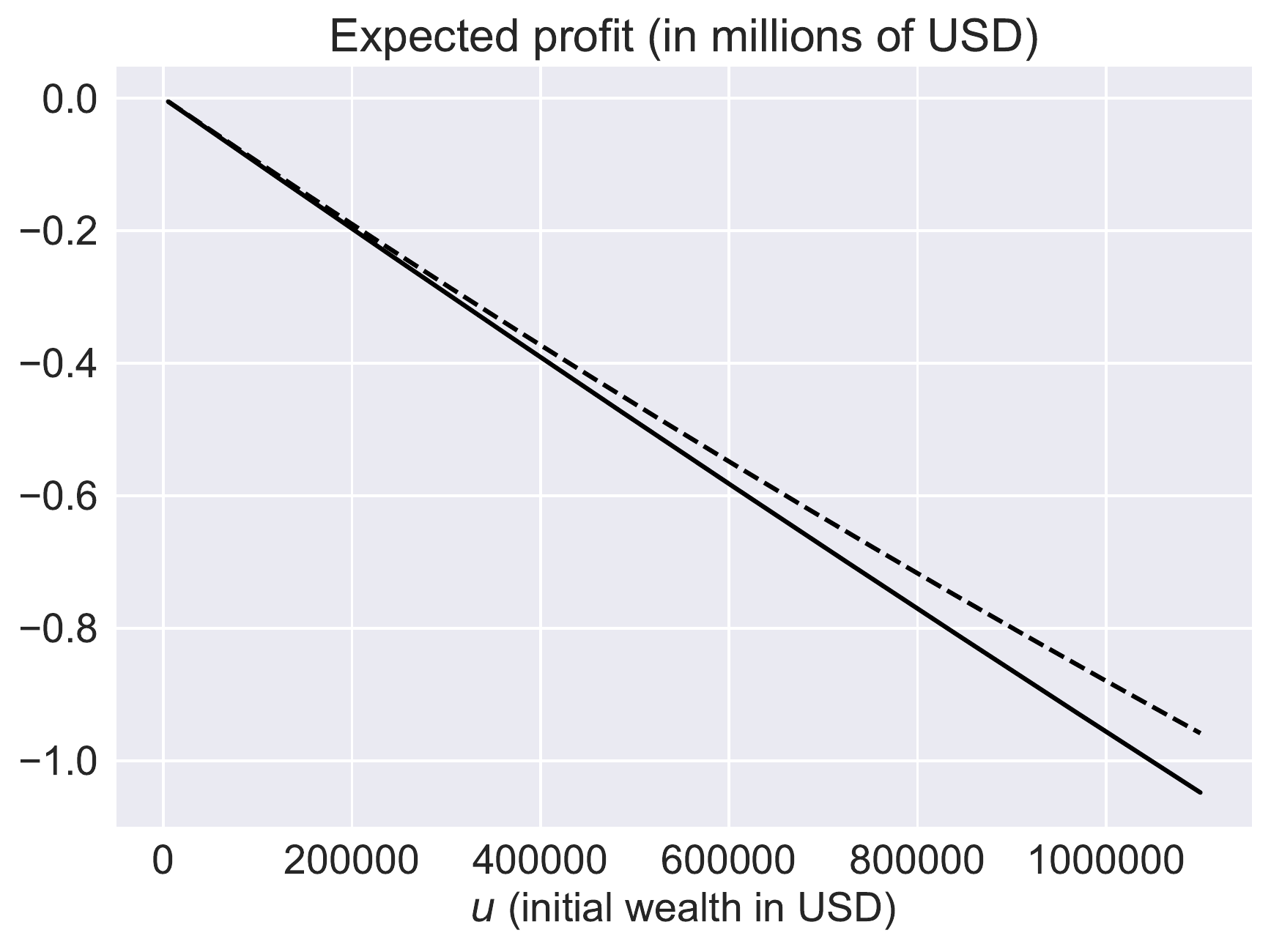}
      \label{sub:rev_difficulty_adjusment_9}
                         }
    \caption{Expected profit over two segments as a function of initial wealth of a miner following the protocol (solid) and a selfish miner (dashed) for various electricity prices with hashpower $p=0.1$ and connectivity $q=0.5$.}
    \label{fig:rev_difficulty_adjusment}
  \end{center}
\end{figure}
 One can observe the different profit and loss profile of a selfish miner compared to that of a miner following the protocol on both segments. If the net profit condition holds when following the protocol, then it also holds for the second segment when blocks were withheld during the first segment. For electricity prices $\pi_W = 0.04, 0.05, 0.06$, the expected profit as a function of $u$ reaches a plateau of level 
\begin{equation}\label{eq:plateau_protocol}
(c-b\lambda p)\times t
\end{equation}
when following the protocol, and 
\begin{equation}\label{eq:plateau_selfish}
(c-b\lambda_2 p)\times t_2,
\end{equation}
when selfish mining is applied during the first segment. The expected profit when following the protocol is greater than the one when applying selfish mining on the first segment for $\pi_W = 0.04,0.05$, see Figures \ref{sub:rev_difficulty_adjusment_4} and \ref{sub:rev_difficulty_adjusment_5}. For $\pi_W = 0.06$, the plateau when following the protocol \eqref{eq:plateau_protocol} is higher than the plateau when withholding blocks \eqref{eq:plateau_selfish}, but the expected profit at lower initial wealth is greater for the selfish miner, see Figure \ref{sub:rev_difficulty_adjusment_6}. For $\pi_W = 0.07$, the net profit condition no longer holds when following the protocol which entails a loss, it holds however on the second segment of the selfish miner, see Figure \ref{sub:rev_difficulty_adjusment_7}. The latter is probably the most desirable situation for a selfish miner. For electricity prices $\pi_W > 0.071$, the net profit condition breaks down in each case, resulting in huge losses for both the selfish and the honest miner (cf.\ Figures \ref{sub:rev_difficulty_adjusment_8} and \ref{sub:rev_difficulty_adjusment_9}). In these cases, we find that selfish mining helps at least in slightly mitigating the losses. This effect seems to fade out as the price of electricity increases.\\

\noindent The above analysis allowed us to distinguish situations where selfish mining can be considered worthwile and when it may not. In particular, it turns out that selfish mining can be advisable when following the protocol is not profitable. 

\section{Conclusion}\label{sec:conc}
In this paper we proposed a risk and profitability analysis framework adapted to the profit and loss profile of blockchain miners. The surplus of the miners is modelled as a stochastic process akin to the risk process in insurance risk theory. In addition to studying the standard ruin probability, we have defined a value function as the expected surplus over an exponentially distributed time horizon. This assumption allowed to work out closed form expressions for the expected profit of a miner that follows the prescribed protocol, and of a miner who seeks to optimize his expected profit by withholding blocks. The explicit solutions enabled a sensitivity analysis with respect to the model parameters and ingredients, and a profitability comparison between the two types of miners.\\

\noindent We find that mining is a business for risk lovers as large expected profits are compensated by great odds of running out of financial resources. While earlier studies suggested that selfish mining can be worthwhile, our analysis gives a quantitative description of the significantly increased risk caused by the delay of capital gains, so that in an analysis that includes the consideration of ruin, selfish mining becomes less attractive.  At the same time, we find that selfish mining can be a reasonable  strategy if following the protocol is not profitable. \\

\noindent The main purpose of the present study was to provide a concrete link between the analysis of mining strategies and modelling tools in applied probability and mathematics. There are of course various directions for future research concerning refinements and extensions of the approach proposed here. It will be interesting in future studies to  account for the variability of the electricity cost over time and the even higher variability of the miner's reward due to the well known volatility of the Bitcoins. The incorporation of transaction fees will also become more relevant in the future as the reward for finding blocks keeps being halved as the number of remaining Bitcoins to be issued declines. In addition, it would be nice to extend the present analysis to incorporate the original selfish mining procedure studied in Eyal and Sirer \cite{EySi18}, and to look into further consequences of selfish mining on the remaining network participants beyond the ones considered here. Finally, the organization of selfish mining in pools may lead to additional features beyond the ones considered in the present approach.  

\appendix
\section{Proof of Theorem \ref{theo:value_Function_Selfish}}\label{Appendix}

	Using the same reasoning as in the proof of Proposition \ref{prop:value_function_protocol} yields the following system of advanced differential equations
	\begin{equation}\label{eq:ODE_system}
	\begin{cases}
	0=&c\vh_0'(u,t)+(\lambda p +1/t)\vh_0(u,t)-\lambda p \vh_1(u,t)-u/t, \\
	0=&c\vh_1'(u,t)+(\lambda +1/t)\vh_1(u,t))-\lambda p \vh_0(u+2b,t) -\lambda(1-p)\vh_{0^\ast}(u,t) -u/t,\\
	0=&c\vh_{0^\ast}'(u,t)+(\lambda +1/t)\vh_{0^\ast}(u,t)-\lambda p \vh_{0^\ast}(u+2b,t) -\lambda(1-p)q\vh_{0^\ast}(u+b,t)\\
	&-\lambda(1-p)(1-q)\vh_{0}(u,t) -u/t.\\
	\end{cases}
	\end{equation}
	Here again the derivative is always with respect to the first argument.
	From the first equation we deduce that 
	\begin{equation}\label{eq:V1_prop}
	\vh_{1}(u,t) = \frac{c}{\lambda p}\vh_{0}'(u,t)+\frac{\lambda p +1/t}{\lambda p}\,\vh_{0}(u,t)-\frac{u}{t\lambda p}.
	\end{equation}
	Inserting \eqref{eq:V1_prop} into the second equation of the system \eqref{eq:ODE_system} yields
	\begin{eqnarray}\label{eq:V_0_ast_prop}
	\vh_{0^\ast}(u,t)& = &\frac{c}{\lambda(1-p)}\vh_{1}'(u,t)+\frac{\lambda +1/t}{\lambda(1-p)}\vh_{1}(u,t)-\frac{\lambda p}{\lambda(1-p)} \vh_{1}(u+2b,t)  \nonumber\\
	&&\quad-\frac{u}{\lambda(1-p)t}\nonumber\\
	&=&{\frac {{c}^{2}}{{\lambda}^{2} \left( 1-p \right) p}}\vh_{0}''(u,t)+{\frac { c \left( 2+
			\left( p+1 \right) \lambda\,t \right) }{{\lambda}^{2} \left( 1-p
			\right) tp}}\vh_{0}(u,t)\nonumber \\
	&&\quad +{\frac { \left( \lambda\,t+1 \right)  \left( \lambda\,t\,p+1 \right)  }{{\lambda}^{2} \left( 1-p \right) {t}^{2}p}
	}\vh_{0}(u,t) - {\frac {p}{1-p}}\vh_{0}(u+2b,t)\nonumber \\
	&&\quad-{\frac {  1+\lambda\,t
			\left( p+1 \right)  }{{\lambda}^{2} \left( 1-p
			\right) {t}^{2}p}}u-{\frac {c}{{\lambda}^{2} \left( 1-p \right) tp}},
	\end{eqnarray}
with boundary conditions $\vh_{0}(0,t) =\vh_{1}(0,t)=\vh_{0^\ast}(0,t)=0$. 	Substituting \eqref{eq:V_0_ast_prop} into the third equation of the system \eqref{eq:ODE_system} then yields an advanced differential equation for $\vh_0(u,t)$ with
	\begin{eqnarray}\label{eq:DDE_Order3}
	0&=& {\frac {{c}^{3}}{{\lambda}^{2} \left( 1-p \right) p}}\vh_0'''\left( u,t
	\right) 
	+{\frac { \left( 3+\lambda\,t
			\left( p+2 \right)  \right) {c}^{2}}{{\lambda}^{2} t\left( 1-p
			\right) p}}\vh_0''(u,t)\nonumber\\
	&&\quad +
	{\frac { \left( \lambda\,t+1 \right) c \left( 3+\lambda\,t
			\left( 2p+1 \right)  \right)}{{\lambda}^{2}{t}^{2} \left(1-p\right) p}}\vh'_0(u,t)\nonumber\\
	&&\quad+{\frac {  1+\lambda\,t \left( p+2 \right) +{\lambda}^{2}{t}^{2} \left( 2p+
			1 \right)+{\lambda}^{3}{t}^{3}p \left(  \left( q-1 \right) {p}^{2}+
			2p\left( 1-q \right) +q \right)    }{{t}^{3}{\lambda}^{2} \left( 1-p \right) p}}\vh_0\left( u,t \right) \nonumber\\
	&&\quad-\lambda\,q\left( 1-p \right)\vh_0\left( u+b,t \right) +{
		\frac { \left( \lambda\,t \left( p-2 \right) -1 \right) p}{ \left( 1-p \right)t }}\vh_0(u+2b,t)
	\nonumber \\
	&&\quad-{\frac {cp }{1-p}}\vh_0'\left( u+2\,b,t \right)-{\frac {  1+
			\lambda\,	t \left( p+2 \right)-{\lambda}^{2}{t}^{2} \left( {p}^{2}-2\,p-1 \right)   }{{t}^{3}{\lambda}^{2} \left( 1-p \right) p}}u\nonumber \\
	&&\quad-{\frac { \left( 2+\lambda\,t \left( p+2 \right)  \right) 
			c}{{\lambda}^{2}{t}^{2} \left( 1-p \right) p}}.
	\end{eqnarray}
	The above boundary conditions now translate to	\begin{eqnarray}\label{eq:boundary}
	\vh_0(0,t) = 0\text{, }\vh_0'(0,t)=0\text{ and }\vh_0''(0,t) = \frac{\lambda^2p^2}{c^2}\vh_0(2b,t)+\frac{1}{ct},\end{eqnarray}
	and from the problem contruction we again have the linear growth condition $\vh_0(u,t)\le k u$ for some $k\ge 0$. This equation constitutes another instance of an advanced functional differential equation, and concerning mathematical considerations of the existence and uniqueness of a solution to it, we refer to the corresponding comments in the proof of Proposition  \ref{prop:value_function_protocol}, where here we even have the additional complication of order 3. We will again construct its solution explicitly and assume it to be unique (which one can again confirm by comparing it to the outcome of a stochastic simulation of $V_0(u)$). 
Concretely, we seek a solution of the form $\vh_0(u,t) =\vh_0^{\text{part}}(u,t)+\vh_0^{\text{hom}}(u,t)$, where $\vh_0^{\text{part}}(u,t)$ is a particular solution of \eqref{eq:DDE_Order3} and $\vh_0^{\text{hom}}(u,t)$ solves the homogeneous equation associated to \eqref{eq:DDE_Order3}
	\begin{eqnarray}\label{eq:DDE_Order3_homogeneous}
	0&=& {\frac {{c}^{3}}{{\lambda}^{2} \left( 1-p \right) p}}\vh_0'''\left( u,t
\right) 
+{\frac { \left( 3+\lambda\,t
		\left( p+2 \right)  \right) {c}^{2}}{{\lambda}^{2} t\left( 1-p
		\right) p}}\vh_0''(u,t)\nonumber\\
		&&\quad +
{\frac { \left( \lambda\,t+1 \right) c \left( 3+\lambda\,t
		\left( 2p+1 \right)  \right)}{{\lambda}^{2}{t}^{2} \left(1-p\right) p}}\vh'_0(u,t)\nonumber\\
&&\quad+{\frac {  1+\lambda\,t \left( p+2 \right) +{\lambda}^{2}{t}^{2} \left( 2p+
		1 \right)+{\lambda}^{3}{t}^{3}p \left(  \left( q-1 \right) {p}^{2}+
		2p\left( 1-q \right) +q \right)    }{{t}^{3}{\lambda}^{2} \left( 1-p \right) p}}\vh_0\left( u,t \right) \nonumber\\
&&\quad-\lambda\,q\left( 1-p \right)\vh_0\left( u+b,t \right) +{
	\frac { \left( \lambda\,t \left( p-2 \right) -1 \right) p}{ \left( 1-p \right)t }}\vh_0(u+2b,t)\nonumber\\
	&&\quad-{\frac {cp }{1-p}}\vh_0'\left( u+2\,b ,t\right).
	\end{eqnarray}
	Inserting a particular solution of the form  $\vh_0^{\text{part}}(u,t) = Bu + C$ in \eqref{eq:DDE_Order3} yields $B = 1$ and 
	\begin{eqnarray*}
		C&=& -\frac{\lambda^{2} {t}^{3} \left\{\lambda pb \left[  \left( q-2 \right) {p
			}^{2}+ \left( 4-2\,q \right) p+q \right] +c \left( {p}^{2}-2\,
			p-1 \right)  \right\}}{{\lambda}^{2}{t}^{2} \left( {
				p}^{2}-2\,p-1 \right) - \left( p+2 \right) \lambda\,t-1}\\
		&&\quad-\frac { \lambda\,t^2\, \left( 2b{p}^{2}\lambda-
				c\left( p+2 \right)  \right)-ct }{{\lambda}^{2}{t}^{2} \left( {
					p}^{2}-2\,p-1 \right) - \left( p+2 \right) \lambda\,t-1}.
	\end{eqnarray*}
	Trying the ansatz $\vh(u,t)=e^{\rho u}$ in the homogeneous equation \eqref{eq:DDE_Order3_homogeneous}, we get the characteristic equation
	\begin{equation}\label{eq:caracteristic_equation}
	\left( {\it D_1}\,\rho+{\it D_2} \right) {{\rm e}^{2\,\rho\,b}}+{\it D_3
	}\,{{\rm e}^{\rho\,b}}={\it D_4}\,{\rho}^{3}+{\it D_5}\,{\rho}^{2}+{\it 
		D_6}\,\rho+{\it D_7},
	\end{equation}
	for $\rho$, where 
	\begin{eqnarray*}
		&&D_1 = {\frac {pc}{1-p}}\geq0\text{, }D_2 = {\frac {p \left( 1+\lambda\,t \left( 2-p \right)  \right) }{t \left( 1-
				p \right) }}\geq0\text{, }D_3 =  \left( 1-p \right) \lambda\,q\geq0,\\
		&&D_4 ={\frac {{c}^{3}}{{\lambda}^{2} \left( 1-p \right) p}}
		\geq0\text{, }\quad
		D_5 = {\frac { {c}^{2}\left( 3+\lambda\,t \left( p+2 \right)  \right) }{t{
					\lambda}^{2} \left( 1-p \right) p}}\geq 0,\\
		&&D_6 = {\frac { \left( 3+\lambda\,t \left( 2p+1 \right)  \right) c
				\left( \lambda\,t+1 \right) }{{t}^{2}{\lambda}^{2} \left( 1-p\right) p}}\geq0\quad\text{and }\\
		&&D_7 = {\frac {1+\lambda\,t \left( p+2 \right)+ {
					\lambda}^{2}{t}^{2}\left( 2p+1 \right)+{\lambda}^{3}{t}^{3}p \left( p(1-q) ( 2-p)+q \right)  }{{\lambda}^{2}{t}^{3}
				\left( 1-p \right) p}}\geq0.
	\end{eqnarray*}
	Note that due to the linear growth condition for $\vh_0(u,t)$ in $u$, any candidate for $\rho$ needs to have negative real part. In fact, in the following lemma we will prove that \eqref{eq:caracteristic_equation} has exactly three such solutions with negative real part. 
	
	\begin{lemma}
	For any choice of parameters $\lambda,c,b,t>0$ and $0<p,q<1$, Equation \eqref{eq:caracteristic_equation} has exactly three solutions in $\rho$ with negative real part, one of which is real-valued.
	\end{lemma}
	\begin{proof}
		Define the functions 
	$$
	f(\rho)={\it D_4}\,{\rho}^{3}+{\it D_5}\,{\rho}^{2}+{\it 
		D_6}\,\rho+{\it D_7}\text{ and }g(\rho)=\left( {\it D_1}\,\rho+{\it D_2} \right) {{\rm e}^{2\,\rho\,b}}+{\it D_3
	}\,{{\rm e}^{\rho\,b}},
	$$ 
	and note that both are holomorphic in the complex plane. 
	We will first show that all three zeros of $f(\rho)$ in the complex plane have negative real part. Define a closed curve $K$ by considering the arc $[-i R,+iR]$ on the imaginary axis together with the semi-circle of radius $R$ to the left of it that is centered in zero, with $R\to\infty$. Going on the semi-circle from $iR$ to $-iR$, the term ${\it D_4}\,{\rho}^{3}$ dominates the value of $f$, and $f(\rho)$ winds 1.5 times around the origin on the way (from $D_4(i x)^3=-i D_4 x^3$ to $D_4(-i x)^3=+i D_4 x^3$ for $x\in{\mathbb R}$). For the part of $K$ on the imaginary axis we have 
	\begin{equation}\label{eq11}
	f(i x)=D_7-D_5x^2+i(D_6 x-D_4x^3), \quad x\in{\mathbb R}.\end{equation}
	As $x$ goes from $-\infty$ to $+\infty$, the real part of \eqref{eq11} becomes positive at $x_1=-\sqrt{D_7/D_5}$ and again negative at $x_2=\sqrt{D_7/D_5}$, and the imaginary part of \eqref{eq11} is negative for $x_1$ and positive for $x_2$ if and only if 
	\[D_4D_7<D_5D_6.\]
	But the latter condition holds for any choice of parameters $\lambda,c,b,t>0$ and $0<p,q<1$, which can be checked with a little algebra (in fact, writing $D_5D_6-D_4D_7$ as a polynomial in the argument $\lambda t$, one can verify that all the coefficients of that polynomial are positive). Due to $D_4(i x)^3=-i D_4 x^3$, this then means that 1.5 further windings are added on the part of $K$ on the imaginary axis, so that altogether we have 3 windings around the origin when going along $K$. By Cauchy's argument principle for holomorphic functions, we hence have established that $f(\rho)$ has three zeros in the negative halfplane. \\

	\noindent In order to extend this result now to $f(\rho)-g(\rho)$ also having exactly three zeros in the negative halfplane, we invoke Rouch\'{e}'s theorem. The latter guarantees the same number of zeroes for $f(\rho)$ and  $f(\rho)-g(\rho)$ inside the closed curve $K$, if 
\begin{equation}\label{eqrouche}\vert -g(\rho)\vert<\vert f(\rho)\vert\end{equation} for all $\rho$ on the curve $K$. On the semi-circle, \eqref{eqrouche} is clear due to the dominance of ${\it D_4}\,{\rho}^{3}$ as $R\to\infty$. This is likewise true for the part of $K$ on the imaginary axis for sufficiently large $R$, so that we are only left to deal with $\rho\in[-iR,iR]$ for not so large $R$. Note that 
\[g(ix)=D_3\cos(xb)+D_2\cos(2xb)-D_1x\sin(2xb)+i(D_1x\cos(2xb)+D_2\sin(2xb)+D_3\sin(xb))\]
for any $x\in{\mathbb R}$. Together with \eqref{eq11}, the condition \eqref{eqrouche} now translates into
\begin{multline}\label{inequ}
(D_1x\cos(2xb)+D_2\sin(2xb)+D_3\sin(xb))^2+(D_3\cos(xb)+D_2\cos(2xb)-D_1x\sin(2xb))^2\\
<(D_7-D_5x^2)^2+(D_6 x-D_4x^3)^2
\end{multline}
for all $x\in{\mathbb R}$. Observe that both sides of \eqref{inequ} are symmetric around $x=0$. The right-hand side of \eqref{inequ} can be written as 
\[D_7^2+(D_6^2 - 2 D_5 D_7)x^2+(D_5^2 - 2 D_4 D_6)x^4+D_4^2 x^6.
\]
and the left-hand side of \eqref{inequ} has the following expansion around zero
\[(D_2+ D_3)^2 + (D_1^2 - 2 b D_1 D_3 - b^2 D_2 D_3)\, x^2 +O(x^4).\]

\noindent In order to show $D_7>D_2+D_3$, one writes $D_7-D_3-D_2$ again as a polynomial in the argument $\lambda t$ and can then verify that all coefficients of that polynomial are positive. In a similar way, one can show that $D_6^2> 2 D_5 D_7$ and $D_5^2 > 2 D_4 D_6$ as well as $D_5^2 - 2 D_4 D_6>D_1^2 - 2 b D_1 D_3 - b^2 D_2 D_3$ for all $b>0$. One then sees that \eqref{inequ} indeed holds for all $x\in{\mathbb R}$, so that \eqref{eqrouche} applies and we have established the existence of exactly three zeros of  \eqref{eq:caracteristic_equation} in the negative halfplane. \\

\noindent It only remains to show that exactly one of these zeros is real-valued. In fact, $g$ is positive and monotone increasing with $
\underset{\rho\rightarrow-\infty}{\lim}g(\rho) = 0$ and $\underset{\rho\rightarrow+\infty}{\lim}g(\rho) = +\infty.
$
On the other hand, with $D_4>0$ we see that $\underset{\rho\rightarrow-\infty}{\lim}f(\rho) = -\infty$ and $\underset{\rho\rightarrow+\infty}{\lim}f(\rho) = +\infty$. The derivative of $f$  
$$
f'(\rho) = 3D_4\rho^2+2D_5\rho+D_6
$$  
is a polynomial of degree $2$ with two negative real roots 
$$
r_1 =-{\frac {1}{ct}} -\frac {\lambda}{c}\text{ and }r_2 = -{\frac {1}{ct}} -\frac {\lambda}{c}{\frac {2\,p+1}{3}}>r_1.
$$
Since
$$
f(r_1) = - \left( 1-q \right)  \left( 1-p\right) \lambda\leq 0\text{ and }f(r_2) = -\,{\frac { \left( -27\,pq+23\,p+4 \right)  \left( 1-p \right) \lambda}{27p}}\leq0,
$$
the function $f$ is negative for $\rho\leq r_2$ and cannot intersect with $g$ there. At the same time,  
$$f(0)-g(0)={\frac {1+{t}^{2} \left(- {p}^{2}+2\,p+1 \right) {\lambda}^{2}+t
		\left( p+2 \right) \lambda}{{\lambda}^{2} \left( 1-p \right) p{t}^{3
}}}\geq0,
$$
which implies the existence of a unique negative  $\rho_1\in\left(r_2,0\right)$ such that $f(\rho_1)=g(\rho_1)$ (there must also exist some real solution  $\rho>0$ to \eqref{eq:caracteristic_equation} as $g$ grows to infinity faster than $f$, but a positive solution is not relevant here).	Since $f(\overline{\rho})=\overline{f(\rho)}$ and  $g(\overline{\rho})=\overline{g(\rho)}$, where $\overline{z}$ denotes the complex conjugate, for every non-real zero $\rho$ its complex conjugate $\overline{\rho}$ also must be a zero, so that the second and third root of \eqref{eq:caracteristic_equation} are complex conjugates $\rho_2\pm i \rho_3$. \end{proof}

	The solution of the homogeneous equation is then of the form 
	\begin{equation}\label{eq:homogeneous_solution}
	\vh_0^{\text{hom}}(u,t)={\it A_1}\,{{\rm e}^{\rho_1\,u}}+{{\rm e}^{\rho_2\,u}} \left[ {\it A_2}\,
	\cos \left( \rho_3\,u \right) +{\it A_3}\,\sin \left( \rho_3\,u \right)\right],
	\end{equation}
	which together with the particular solution $\vh_0^{\text{part}}$ provides the generic solutions of the equation \eqref{eq:DDE_Order3} with
	\begin{equation}\label{eq:general_solution}
	\vh_0(u,t)={\it A_1}\,{{\rm e}^{\rho_1\,u}}+{{\rm e}^{\rho_2\,u}} \left[ {\it A_2}\,
	\cos \left( \rho_3\,u \right) +{\it A_3}\,\sin \left( \rho_3\,u \right)\right] + bu +C.
	\end{equation}
	Finally, the boundary conditions \eqref{eq:boundary} lead, after some algebra, to the system of equations \eqref{system} for the determination of $A_1,A_2$ and $A_3$.\hfill $\Box$

\section{Proof of Corollary \ref{coro2}}\label{AppendixB}
The corollary follows indeed by a suitable adaptation of the proof of Theorem \ref{theo:value_Function_Selfish}. Denote by $\ph_z(u,t)$ the ruin probability up to the finite exponential time horizon (with mean $t$) when starting in state $z$ ($z\in\{0,1,0^\ast\}$). Adapting the above arguments to the case of the ruin probability, the system of advanced differential equations in this case is
 \begin{equation}\label{eq:ODE_system_ruin}
 \begin{cases}
 0=&c\,\ph_0'(u,t)+(\lambda p +1/t)\,\ph_0(u,t)-\lambda p \,\ph_1(u,t), \\
 0=&c\,\ph_1'(u,t)+(\lambda +1/t)\,\ph_1(u,t)-\lambda p \,\ph_0(u+2b,t) -\lambda(1-p)\,\ph_{0^\ast}(u,t) ,\\
 0=&c\,\ph_{0^{\ast}}'(u,t)+(\lambda +1/t)\,\ph_{0^{\ast}}(u,t)-\lambda p \,\ph_0(u+2b,t) -\lambda(1-p)q\,\ph_{0}(u+b,t)\\
 &\quad-\lambda(1-p)(1-q)\,\ph_{0}(u,t) ,\\
 \end{cases}
 \end{equation}
which is exactly the homogeneous version of\eqref{eq:ODE_system} (again the derivative is always with respect to the first argument). However, now the boundary conditions are $\,\ph_0(0,t)=\,\ph_1(0,t)=\,\ph_{0^{\ast}}(0,t)=1$, as in each state ruin is immediate when starting without initial capital. Analogously to above this translates into 
a differential equation of order 3 with advanced argument for $\ph_0(u,t)$, namely
\begin{eqnarray}\label{eq:DDE_Order3ruin}
0&=& {\frac {{c}^{3}}{{\lambda}^{2} \left( 1-p \right) p}}\,\ph_0'''\left( u,t
\right) 
+{\frac { \left( 3+\lambda\,t
		\left( p+2 \right)  \right) {c}^{2}}{{\lambda}^{2} t\left( 1-p
		\right) p}}\,\ph_0''(u,t)\nonumber\\
&&\quad+
{\frac { \left( \lambda\,t+1 \right) c \left( 3+\lambda\,t
		\left( 2p+1 \right)  \right)}{{\lambda}^{2}{t}^{2} \left(1-p\right) p}}\,\ph'_0(u,t)\nonumber\\
&&\quad +{\frac {  1+\lambda\,t \left( p+2 \right) +{\lambda}^{2}{t}^{2} \left( 2p+
		1 \right)+{\lambda}^{3}{t}^{3}p \left(  \left( q-1 \right) {p}^{2}+
		2p\left( 1-q \right) +q \right)    }{{t}^{3}{\lambda}^{2} \left( 1-p \right) p}}\,\ph_0\left( u,t \right) \nonumber\\
&&\quad-\lambda\,q\left( 1-p \right)\,\ph_0\left( u+b,t \right) +{
	\frac { \left( \lambda\,t \left( p-2 \right) -1 \right) p}{ \left( 1-p \right)t }}\,\ph_0(u+2b,t)\nonumber\\
	&&\quad-{\frac {cp }{1-p}}\,\ph_0'\left( u+2\,b,t \right).\nonumber
\end{eqnarray}
The latter corresponds to the homogeneous equation of \eqref{eq:DDE_Order3}. However, the boundary conditions are now slighty different and amount to 
\begin{eqnarray}\label{eq:boundaryruin}
\ph_0(0,t) = 1\text{, }\ph_0'(0,t)=-\frac{1}{ct}\text{ and }\ph_0''(0,t) = \frac{\lambda^2p^2}{c^2}\,(\ph_0(2b,t)-1)+\frac{1}{c^2t^2}.\end{eqnarray}
Furthermore, we know that $\lim_{u\to\infty}\ph_0(u,t)=0$. Like in the proof of Theorem \ref{theo:value_Function_Selfish}, the solution can then be found to be of the form
 	\begin{equation}\label{eq:homogeneous_solution}
 \ph_0(u,t)={\it C_1}\,{{\rm e}^{\rho_1\,u}}+{{\rm e}^{\rho_2\,u}} \left[ {\it C_2}\,
 \cos \left( \rho_3\,u \right) +{\it C_3}\,\sin \left( \rho_3\,u \right)\right],
 \end{equation}
 where the linear system \eqref{systemruin} for the constants $C_1,C_2,C_3$ is finally derived from \eqref{eq:boundaryruin} by some elementary calculations.   
\hfill $\Box$

\section*{Acknowledgements}  
The authors thank Hans Gerber for helpful discussions concerning the presentation of the manuscript. Hans\-j\"{o}rg Albrecher acknowledges financial support from the Swiss National Science Foundation Project 200021\_191984. Pierre-O. Goffard acknowledges the support of l'Université Claude Bernard Lyon 1  via the Bonus Qualité Recherche (BQR) starting grant. The code to reproduce all the plots from the paper are available on the following Github repository \url{https://github.com/LaGauffre/Blockchain_mining_risk_analysis}.

\bibliography{Albrecher_Goffard_Selfish_MineRb}
\bibliographystyle{plain}

\end{document}